\documentclass[aps,pre,reprint,groupedaddress,floatfix,amsmath,amssymb]{revtex4-2}
\usepackage[utf8]{inputenc}
\usepackage{amsfonts,newlfont,amscd,amsthm,amsxtra,mathrsfs,nicefrac,pifont}
\usepackage{bold-extra}
\usepackage{booktabs}
\usepackage{longtable}
\usepackage{graphicx,url}
\usepackage[bookmarks=false,colorlinks=true,urlcolor=blue,linkcolor=blue,citecolor=blue]{hyperref}
\newcommand{\SE}[1]{\mathbf{#1}}

\DeclareMathOperator{\mft}{\mathfrak{t}}
\DeclareMathOperator{\bmft}{\bar{\mathfrak{t}}}
\DeclareMathOperator{\mfs}{\mathfrak{s}}
\DeclareMathOperator{\bmfs}{\bar{\mathfrak{s}}}
\DeclareMathOperator{\mfa}{\mathfrak{a}}
\DeclareMathOperator{\bmfa}{\bar{\mathfrak{a}}}

\DeclareMathOperator{\bfa}{\mathbf{a}}
\DeclareMathOperator{\bfb}{\mathbf{b}}
\DeclareMathOperator{\mP}{\mathbb{P}}
\DeclareMathOperator{\RR}{\mathbb{R}}

\DeclareMathOperator{\Mod}{mod}
\DeclareMathOperator{\bit}{bit}

\DeclareMathOperator{\red}{red}
\DeclareMathOperator{\vk}{VK}
\DeclareMathOperator{\UI}{UI}
\DeclareMathOperator{\sx}{sx}

\newtheorem{theorem}{Theorem}[section]
\newtheorem{proposition}{Proposition}[section]

\newtheorem{lemma}{Lemma}[section]

\newtheorem{axiom}{Axiom}[section]

\begin{document}

\title{Introducing a differentiable measure of pointwise shared information}

\author{Abdullah Makkeh}
\email[]{abdullah.alimakkeh@uni-goettingen.de}
\author{Aaron J. Gutknecht}
\email[]{agutkne@uni-goettingen.de}
\altaffiliation[Also at ]{MEG Unit, Brain Imaging Center, Goethe University, Frankfurt, Germany}
\author{Michael Wibral}
\email[]{michael.wibral@uni-goettingen.de}
\affiliation{Campus Institute for Dynamics of Biological Networks, Georg-August Univeristy, Goettingen, Germany}

\date{\today}

\begin{abstract}
Partial information decomposition (PID) of the multivariate mutual information describes the distinct ways in which a set of source variables contains information about a target variable. The groundbreaking work of Williams and Beer has shown that this decomposition cannot be determined from classic information theory without making additional assumptions, and several candidate measures have been proposed, often drawing on principles from related fields such as decision theory. None of these measures is differentiable with respect to the underlying probability mass function. We here present a novel measure that satisfies this property, emerges solely from information-theoretic principles, and has the form of a local mutual information. We show how the measure can be understood from the perspective of exclusions of probability mass, a principle that is foundational to the original definition of the mutual information by Fano. Since our measure is well-defined for individual realizations of the random variables it lends itself for example to local learning in artificial neural networks. We also show that it has a meaningful Moebius inversion on a redundancy lattice and obeys a target chain rule. We give an operational interpretation of the measure based on the decisions that an agent should take if given only the shared information.
\end{abstract}


\maketitle

\section{Introduction}
 What are the distinct ways in which a set of source variables may contain information about a target variable? How much information do input variables provide \textit{uniquely} about the output, such that this information about the output variable cannot be obtained by any other input variable, or collections thereof? How much information is provided in a \textit{shared} way, i.e., redundantly, by multiple input variables, or multiple collections of these? And how much information about the output is provided \textit{synergistically} such that it can only be obtained by considering many or all input variables together? Answering questions of this nature is the scope of partial information decomposition (PID). 
 
A solution to this problem has been long desired in studying  complex systems~\cite{brenner2000adaptive,latham2005synergy,margolin2006aracne} but seemed out of reach until the groundbreaking study of Williams and Beer~\cite{williams2010nonnegative}. This study provided first insights by establishing that information theory is lacking axioms to uniquely solve the PID problem. Such axioms have to be chosen in a way that satisfies our intuition about shared, unique, and synergistic information (at least in simple corner cases). However, further studies in~\cite{bertschinger2013shared,harder2013bivariate} quickly revealed that not all intuitively desirable properties, like positivity, zero redundant information for statistically independent input, a chain rule for composite output variables, etc., were compatible, and the initial measure proposed by Williams and Beer was rejected on the grounds of not fulfilling certain desiderata favored in the community. Nevertheless, the work of Williams and Beer clarified that indeed an axiomatic approach is necessary and also highlighted the possibility that the higher order terms (or questions) that arose when considering more than two input variables could be elegantly organized into contributions on the lattice of antichains (see more below). Approaches that do not fulfill the Williams and Beer desiderata have been suggested, e.g., \cite{quax2017quantifying,perrone2016hierarchical}. However, these approaches fail to quantify all the desired quantities and, therefore, answer a question different from that posed by PID.

Subsequently, multiple PID frameworks have been proposed, and each of them has merits in the application case indicated by its operational interpretation (Bertschinger et al.~\cite{bertschinger2014quantifying}, e.g., justify their measure of unique information in a decision-theoretic setting).  However, all measures lacked the property of being well defined on individual realizations of inputs and outputs (localizability), as well as continuity and differentiability in the underlying joint probability distribution.  These properties are key desiderata for the settings of interest to neuroscientists and physicists, e.g.,  for distributed computation, where locality is needed to unfold computations in space and time ~\cite{lizier2012localstorage,schreiber2000measuring, wibral2013measuring,lizier2008localtransfer}; for learning in neural networks~\cite{wibral2017partialneuralgoal,kay2011coherent} where differentiability is needed for gradient descent and localizability for learning from single samples and minibatches; for neural coding~\cite{wibral2015bitsbrains,kay2011coherent} where localizability is important to evaluate the information value of individual inputs that are encoded by a system; and for problems from the domain of complex systems in physics as discussed in~\cite{deco2012information}.

While the first two properties have very recently been provided by the pointwise partial information decomposition (PPID) of Finn and Lizier~\cite{finn2018pointwise}, differentiability is still missing, as is the extension of most measures to continuous variables. Differentiability, however, seems pivotal to exploit PID measures for learning in neural networks -- as suggested for example in~\cite{wibral2017partialneuralgoal}, and also in physics problems.

Therefore, we here rework the definition of Finn and Lizier~\cite{finn2018pointwise} in order to define a novel PID measure of shared mutual information that is localizable and also differentiable. We aim for a measure that adheres as closely as possible to the original definition of (local) mutual information -- in the hope that our measure will inherit most of the operational interpretation of local mutual information. We also seek to avoid invoking assumptions or desiderata from outside the scope of information theory, e.g., we explicitly seek to avoid invoking desiderata from decision or game theory. We note that adhering as closely as possible to information-theoretic concepts should also simplify finding localizable and differentiable measures.

Our goals above suggest that we have to abandon positivity for the parts (called atoms in ~\cite{williams2010nonnegative}) of the decomposition, simply because the local mutual information can be already negative \footnote{This can be seen as follows: Assuming that the negative local MI consists only of shared information, then this local shared information must be negative, enforcing the existence of negative local shared information. Now assuming that this shared information does not differ from realization to realization -- something we should consider possible at this point -- while the other contributions vary, then this leads to a shared information that is also negative on average, also see \cite{finn2018pointwise}} With respect to a negative shared information in the PID we aim to preserve the interpretation of negative terms as being misinformative, in the sense that obtaining negative information will make a rational agent more likely to make the wrong prediction about the value of a target variable. Our goals also strongly suggest to avoid computing the minimum (or maximum) of multiple information expressions anywhere in the definition of the measure. This is because taking a minimum or maximum would almost certainly collide with differentiability and also a later extension to continuous variables.

The paper proceeds as follows. First, Section~\ref{sec:Logic}, introduces our measure of shared information $i_\cap^{\sx}$. Then, section~\ref{sec:prob_exc_shrd_info} lays out how $i_\cap^{\sx}$ can be understood based on the concept of shared probability mass exclusions. Section~\ref{sec:lattice} utilizes $i_\cap^{\sx}$ to obtain a full PID and establishes its differentiability. Then, Section~\ref{sec:discussion} discusses some implications of $i_\cap^{\sx}$ being a local mutual information, its operational interpretation, and some key applications of $i_\cap^{\sx}$. Finally, Section~\ref{sec:examples} concludes by several examples. 

\section{\label{sec:Logic} Definition of the measure $i_\cap^{sx}$ of pointwise shared information}
We begin by considering discrete random variables $S_1,\ldots,S_n$ and $T$ where the $S_i$ are called the sources and $T$ is the target. Suppose now that these random variables have taken on particular realizations $s_1,\ldots,s_n$ and $t$. Our goal is to quantify the \textit{pointwise} shared information that the source realizations carry about the target realization. We will proceed in three steps: (1) we define the information shared by \textit{all} source realizations about the target realization, (2) we define pointwise shared information for any \textit{subset} of source realizations, and (3) we provide the complete definition of the information shared by \textit{multiple subsets} of source realizations. 

So how much information about the target realization $t$ is redundantly contained in all source realizations $s_i$? We propose that this information can be quantified as the information about the target realization provided by the truth of the statement 
\begin{equation}
\mathcal{W}_{s_1,\ldots,s_n} = \big((S_1=s_1) \vee \ldots \vee  (S_n=s_n)\big)
\end{equation}
i.e., by the inclusive OR of the statements that each source variable has taken on its specific realization. This information in turn can be understood as a regular pointwise mutual information between the target realization $t$ and the indicator random variable~\footnote{Note that the idea of using an auxiliary random variable ($\mathbf{I}_{\mathcal W}$ in our case) is not novel per se. Quax et al.~\cite{quax2017quantifying} has defined synergy using auxiliary random variable. However, their auxiliary random variable is conceptually different from $\mathbf{I}_{\mathcal W}$ and their approach yielded a `stand-alone' measure of synergistic information without providing any decomposition.} of the statement $\mathcal{W}_{s_1,\ldots,s_n} $ assuming the value 1:
\begin{align}\label{eq:logic_shr}
    i_\cap^{\sx}(t: s_1; \ldots;s_n) 
    &:=\log_2\frac{p(t\mid\mathbf{I}_{\mathcal{W}_{s_1,\ldots,s_n}}=1)}{p(t)}\\
    &=\log_2\frac{p(t\mid\mathcal{W}_{s_1,\ldots,s_n}=\text{true})}{p(t)}.
\end{align}
The superscript ``sx'' stands for ``shared exclusion'' and will be explained in more detail in the next section. The reason for the choice of $\mathcal{W}_{s_1,\ldots,s_n}$ is the following: the truth of this statement can be verified by knowing the realization of \textit{any} single source variable, i.e., knowing that $S_i=s_i$ for at least one $i$. Thus, whatever information can be obtained from $\mathcal{W}_{s_1,\ldots,s_n}$ can also be obtained from any individual statement $S_i=s_i$. In other words, the statement $\mathcal{W}_{s_1,\ldots,s_n}$ \textit{only} contains information that is redundant to all source realizations. Conversely, whatever information can be obtained from all individual statements $S_i = s_i$ can also be obtained from $\mathcal{W}_{s_1,\ldots,s_n}$ because it implies that at least one of the statements $S_i = s_i$ has to be true. In other words, \textit{all} of the information shared by the source realizations is contained in the statement $\mathcal{W}_{s_1,\ldots,s_n}$. Accordingly, the statement $\mathcal{W}_{s_1,\ldots,s_n}$ exactly captures the information redundantly contained in the source realizations. Any logically stronger or weaker statement would either contain some nonredundant information or miss out on some redundant information respectively. For a more comprehensive and foundational version of this argument, connecting principles from mereology (the study of parthood relations) and formal logic,  see \cite{gutknecht2020bits}.

Now, this definition is not entirely complete yet since it only quantifies the information shared by \textit{all} source realizations $s_1,\ldots,s_n$. However, a full-fledged measure of shared information also has to specify the information shared by (1) any \textit{subset} of source realizations (e.g.,the information shared by $s_1$ and $s_3$) and (2) \textit{multiple subsets} of source realizations (e.g., the information shared by $(s_1,s_2)$ and $(s_2,s_3)$) \cite{williams2010nonnegative}. The definition for a subset $\bfa \subseteq \{1,\ldots,n\}$ is straightforward: the information shared by the corresponding realizations $(s_i \mid i \in \bfa)$ is the information provided by the statement
\begin{equation}
\mathcal{W}_{\bfa} = \left(\bigvee_{i \in \bfa} S_i = s_i\right)
\end{equation}
i.e., by the logical OR of statements $S_i = s_i$ \textit{where $i$ is in the subset in question}. Note that in the following we will refer to sets of source realizations \textit{by their index sets} for brevity. So we will generally say ``the set of source realizations $\bfa$'' instead of ``the source realizations  $(s_i \mid i \in \bfa)$''. There are formal reasons why it is preferable to work with index sets that will become apparent in Section \ref{sec:lattice}.

Now, how about the case of multiple subsets? Note first that the pointwise mutual information provided by a given subset $\bfa$ of source realizations about the target realization is the information provided by the logical AND of the corresponding statements $S_i = s_i$:
\begin{equation}
i\left(t: (s_i)_{i \in \bfa} \right) = \log_2\frac{p(t\mid (\bigwedge_{i\in \bfa} S_i = s_i ) =\text{true})}{p(t)}.
\end{equation}
Accordingly, the information shared by multiple subsets of source realizations $\bfa_1,\ldots,\bfa_m$ can be quantified as the information provided by the logical OR of the associated logical AND statements, i.e., as the information provided by the statement
\begin{equation}
\mathcal{W}_{\bfa_1,\ldots,\bfa_m} = \left(\bigvee_{i=1}^{m} \bigwedge_{j\in \bfa_i} S_j = s_j \right).
\end{equation}
The underlying reasoning is exactly as described above: whatever information can be obtained from the $\mathcal{W}_{\bfa_1,\ldots,\bfa_m}$ can also be obtained from all of the conjunctions $\bigwedge_{j\in \bfa_i} S_j = s_j$ because as soon as the truth of one of the conjunctions is known the truth of $\mathcal{W}_{\bfa_1,\ldots,\bfa_m}$ is known as well. Conversely, whatever information can be obtained from all conjunctions can also be obtained from $\mathcal{W}_{\bfa_1,\ldots,\bfa_m}$ since this statement implies that at least one conjunction must be true. This leads us to the final definition of the information shared by arbitrary subsets of source realizations $\bfa_1,\ldots,\bfa_m$:
\begin{align}
i_\cap^{\sx}(t: \bfa_1; \ldots;\bfa_m) 
    &:=\log_2\frac{p(t\mid\mathbf{I}_{\mathcal{W}_{\bfa_1,\ldots,\bfa_n}}=1)}{p(t)}\\
    &=\log_2\frac{p(t\mid\mathcal{W}_{\bfa_1,\ldots,\bfa_n}=\text{true})}{p(t)}.
\end{align}
Note that this general definition agrees with the above definition of the information shared by all source realizations or subsets thereof. We would also like to emphasize here again that $i_\cap^{\sx}$ has the form of a local mutual information. This feature is of particular importance in the following section where we aim to provide further intuition for the measure by showing that it can also be motivated from the perspective of probability mass exclusions as discussed in \cite{finn2018probability}. 

\section{\label{sec:prob_exc_shrd_info}Shared mutual information from shared exclusions of probability mass}
Shannon information can be seen as being induced by exclusion of probability mass (e.g, \cite[Sec.~2.1.3]{wibral2015bitsbrains}), and the same perspective can actually be applied to the mutual information as well -- as explicitly derived by Finn and Lizier~\cite{finn2018probability}. In our approach to shared information, we suggest to keep intact this central information-theoretic principle that binds the exclusion of probability mass to information and mutual information. We now first review the probability exclusion perspective on local mutual information. Subsequently, we show how the measure $i_\cap^{\sx}$ of shared information, itself being a local mutual information, can be motivated from the same perspective as well.

\subsection{\label{subsec:mi_exclusion}Mutual information from exclusions of probability mass}
The local mutual information~\cite{fano1961transmission} obtained from a realization $(t,s)$ of two random variables $T$ and $S$ is
\begin{equation}\label{eq:localmi}
    i(t:s)=\log_2\frac{p(t\mid s)}{p(t)}.
\end{equation}
\noindent 
This means that $i(t:s)$ compares the probability of observing $t$ after observing $s$ to the prior $p(t).$ Thus, $s$ is said to be \emph{informative} (resp.~\emph{misinformative}) about $t$ if the chance of $t$ occurring increases (resp.~decreases) after observing $s$ compared to the prior probability $p(t)$, i.e., if $i(t: s)>0$ (resp.~$i(t:s)<0$).

The definition of $i(t:s)$ can be understood in terms of excluding certain probability mass \cite{finn2018probability} by rewriting it as

\begin{equation}\label{eq:lmi_sets}
    i(t,s)=\log_2\frac{\mP( \mathfrak{t})-\mP( \mathfrak{t}\cap \mathfrak{\bar s})}{1-\mP( \mathfrak{\bar s})} - \log_2{\mP( \mathfrak{t}})~,
\end{equation}
\noindent where $\bar{\mathfrak{s}}$ is the set complement of the event $\mathfrak{s} = \{S=s\}$ and $\mathfrak{t} = \{T=t\}$. Looking at it in this way, pointwise mutual information can be conceptualized as follows (illustrated in FIG.~\ref{fig:remove_rescale}): (i) ``removing'' all points from the initial sample space $\Omega$ that are incompatible with the observation of a specific $s$ by giving them measure zero--for the event $\mathfrak{t}$ this has the consequence that a part of it is also removed, i.e., $\mP( \mathfrak{t})-\mP( \mathfrak{t}\cap \mathfrak{\bar s})$; (ii) \emph{rescaling} the probability measure to again have properly normalized probabilities, i.e., dividing by $1-\mP(\mathfrak{\bar s})$; and (iii) \emph{comparing} the size of $\mft$ after observing $s$ to the prior $\mP(\mft)$ on a logarithmic scale. The remove-rescale procedure is a conceptual way of thinking about the changes to $\Omega$ (after observing $s$) that are reflected in the conditional measure $\mP(\cdot\mid\mfs)$.

This derivation of local mutual information can be generalized to any number of sources. For instance, the joint local mutual information of $s_1,s_2$ about $t$ is

\begin{equation} \label{eq:lmi_n_var}
    i(t:s_1,s_2)=\log_2\frac{\mP( \mathfrak{t})-\mP( \mathfrak{t}\cap(\bar{\mfs}_1\cup\bar{\mfs}_2)) }{1-\mP( \bar{\mfs}_1\cup\bar{\mfs}_2)} - \log_2{\mP( \mathfrak{t}}).
\end{equation}
The two conserved key principles here are that (i) the mutual information is always induced by exclusion of the probability mass related to events that are impossible after the observation of $s_1,\ldots, s_n$, i.e., $\bar{\mathfrak{s}}_1, \ldots, \bar{\mathfrak{s}}_n$, and (ii) the probabilities are rescaled by taking into account these very same exclusions. These core information-theoretic principles can be utilized to motivate the measure $i_\cap^{\sx}$ as explained in the next section.
\begin{figure}
    \centering
    \includegraphics[width=0.4\textwidth, angle=0]{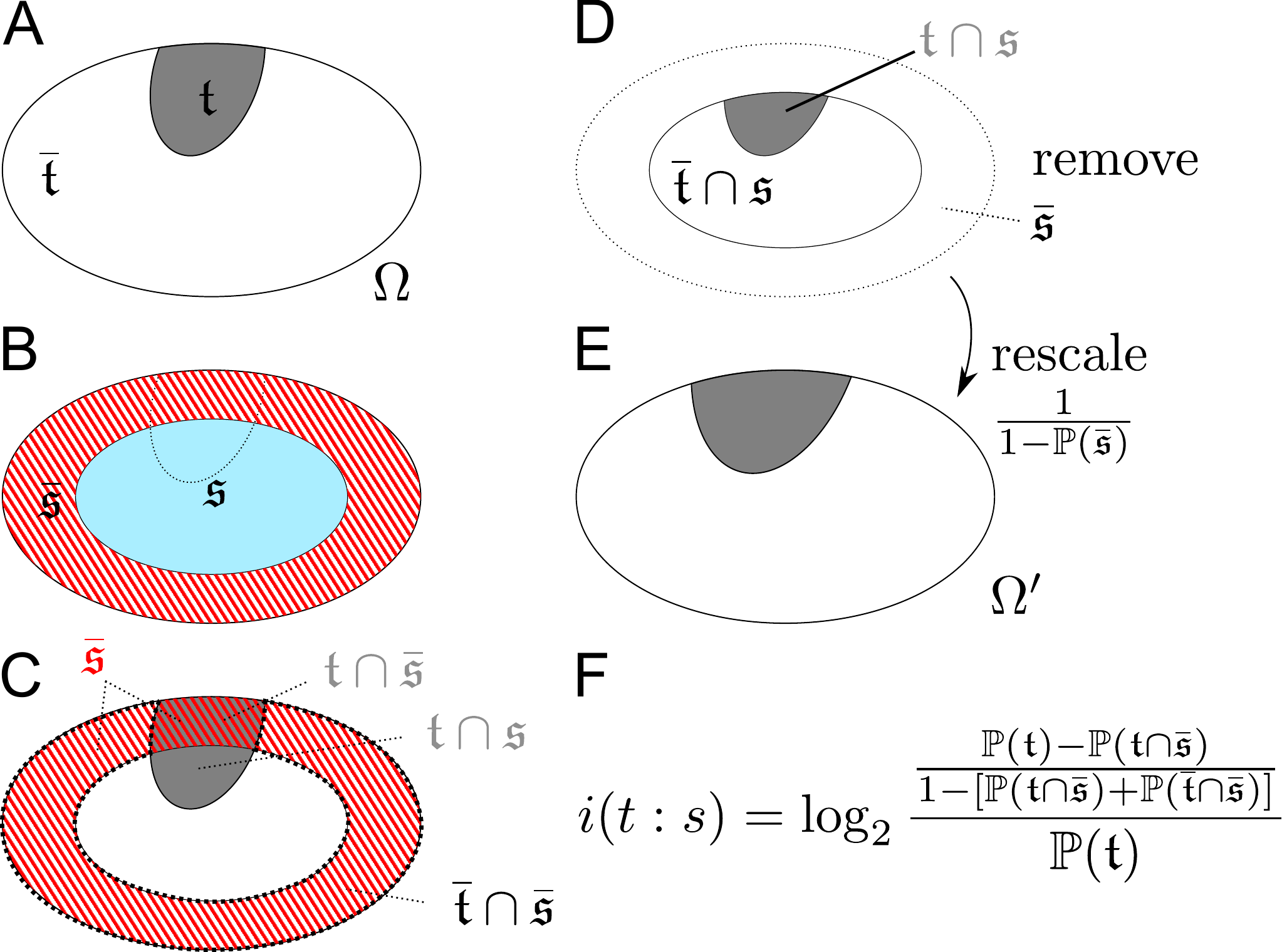}
    \caption{\label{fig:remove_rescale}\textbf{Depiction of deriving the local mutual information $i(t:s)$ by excluding the probability mass of the impossible event $\bmfs$ after observing $\mfs$.} (A) Two events $\mft, \bmft$ partition the sample space $\Omega$. (B) Two event partition $\mfs, \bmfs$ of the source variable $S$ in the sample space $\Omega$.  The occurrence of $\mfs$ renders $\bmfs$ impossible (red (dark gray) stripes). (C) $\mft$ may intersect with $\mfs$ (gray region) and $\bmfs$ (red (dark gray) hashed region). The relative size of the two intersections determines whether we obtain information or misinformation, i.e.~whether $\mft$ becomes relatively more likely after considering $\mfs$, or not (D), considering the necessary rescaling of the probability measure (E). Note that if the gray region in (E) is larger (resp.~smaller) than that in (A), then $s$ is informative (resp.~misinformative) about $t$ since observing $s$ hints that $t$ is more (reps.~less) likely to occur compared to an ignorant prior. (F) shows why the misinformative exclusion $\mP(\mft\cap\bmfs)$ (intersection of red (dark gray) hashes with gray region) cannot be cleanly separated from the informative exclusion, $\mP(\bmft\cap \bmfs)$ (dotted outline in (C)), as stated already in~\cite{finn2018probability}. This is because these overlaps appear together in a sum inside the logarithm, but this logarithm in turn guarantees the additivity of information terms. Thus the additivity of (mutual) information terms is incompatible with an additive separation of informative and misinformative exclusions \emph{inside} the logarithms of the information measures.}
\end{figure}

\subsection[]{\label{subsec:def_isx}$i^{\sx}_{\cap}$ from shared exclusions of probability mass}
The core idea is now that just as mutual information is connected to the exclusion of probability mass, \textit{shared} information should be connected to \textit{shared} exclusions of probability mass, i.e., to possibilities being excluded redundantly by all (joint) source realizations in question. Now, what is excluded by a given joint source realization $\bfa_j$ is precisely the complement of the event $\mathfrak{a}_j = \bigcap_{i \in \bfa_j} \{S_i = s_i\}$. Thus, to evaluate the information shared by the joint source realizations $\bfa_1,\ldots, \bfa_m$, we need to remove and rescale by the intersection of the complement events $\bar{\mathfrak{a}}_j$. This intersection contains points that are excluded by all joint source realizations in question. Hence, we arrive at
\begin{equation}
\small
\begin{split}\label{eq:isx}
i^{\sx}_{\cap}(t:\bfa_1;\bfa_2;\ldots; \bfa_n)   &:=\log_2\frac{\mP( \mathfrak{t})-\mP( \mft\cap(\bmfa_1 \cap \bmfa_2\cap \ldots \cap \bmfa_n) )}{1-\mP(\bmfa_1 \cap \bmfa_2\cap \ldots \cap \bmfa_n) }\\ 
                                        &- \log_2 \mP(\mathfrak{t}).
\end{split}
\end{equation}

\begin{figure}
    \centering
    \includegraphics[width=0.39\textwidth]{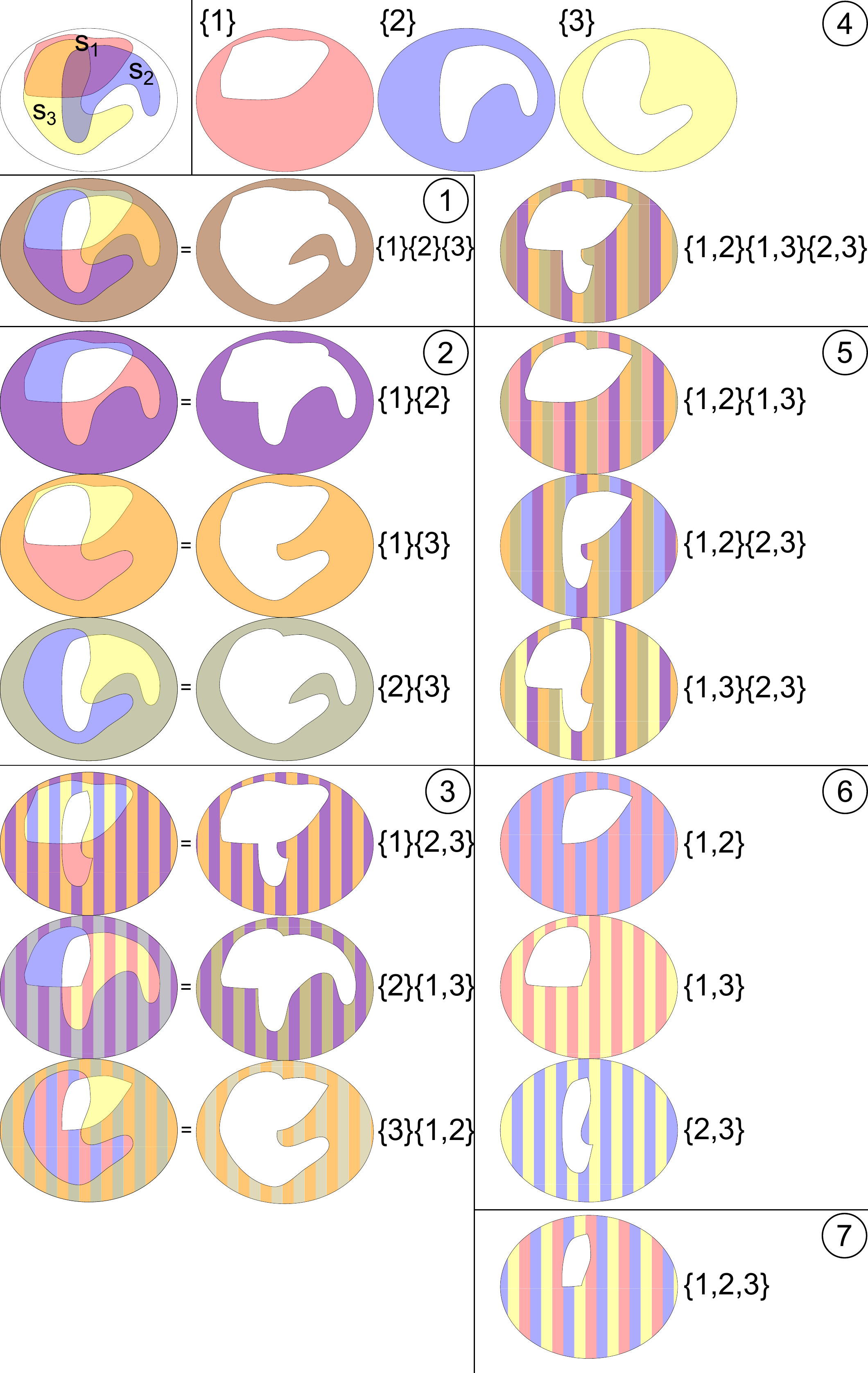}
    \caption{\label{fig:Exclusion_lattice}\textbf{Shared exclusions  in the three-source variable case.} \emph{Upper left:} A sample space with three events $\mfs_1$, $\mfs_2$, $\mfs_3$ from three source variables (their complements events are depicted in (4)). For clarity, $\mft$ is not shown, but may arbitrarily intersect with any intersections/unions of $\mfs_i$. The remaining panels show the induced exclusions by different combinations of $\bfa_i$. These exclusions arise by taking the corresponding unions and intersections of sets. Which unions and intersections were taken can be deduced by the shapes of the remaining, nonexcluded regions. For (1)-(3) we show the shared exclusions for combination of singletons ((1) and (2)) and those of singletons and coalitions, such as the events of the collections (\emph{left}) and the shared exclusions (\emph{right}). For (4)-(7) we only show shared exclusions.
    The online version uses the additional, nonessential color-based mark-up of unions and intersections:  An \emph{intersection exclusion} is indicated by the \emph{mix} of the individual colors, e.g., the $\{1\}\{2\}$ exclusion is $\bmfs_1\cap\bmfs_2$ and mixes red and blue to purple, and a \emph{union exclusion} is indicated by a \emph{pattern} of the individual colors, e.g., the $\{1,2\}$ exclusion is $\bmfs_1\cup\bmfs_2$ and takes a red-blue pattern.}
\end{figure}

It is straightforward to show that this definition coincides with the one given in Section \ref{sec:Logic}. FIG~\ref{fig:Exclusion_lattice} depicts all possible exclusions in the case of three sources. This concludes our exposition of the measure of shared information $i_\cap^{\sx}$. In the next section, we show how this measure induces a meaningful and differentiable partial information decomposition.

\section{\label{sec:lattice}Lattice structure and Differentiability}
We now present a lattice structure that yields a pointwise partial information decomposition (PPID) when endowed with $i^{\sx}_\cap$ and show that all of the resulting PPID terms are differentiable. The lattice structure was originally introduced by Williams and Beer \cite{williams2010nonnegative} on the basis of a range of axioms they placed on the concept of redundant information (see below). As we showed in \cite{gutknecht2020bits} it can also be derived from elementary parthood relationships between the PID terms (also called PID atoms) and mutual information terms. 

\subsection{Lattice structure} 
Williams and Beer in their seminal work~\cite{williams2010nonnegative} showed that in order to capture all the information contributions that a set of sources has about a target, we need to look at the level of \emph{collections of sources}. That is, each combination of collections of sources captures a PPID term (an information contribution / information atom). Their argument was based on an analysis of the concept of redundant information, i.e., the information shared by multiple collections of sources. In particular, they argued that any measure of shared information should satisfy certain desiderata, referred to as W\&B axioms (see Axioms~\ref{ax:sym},~\ref{ax:mono}, and~\ref{ax:self}). These axioms imply that the domain of the shared information function can be restricted to the \emph{antichain combinations}, i.e., any combination of collections of sources such that none of the collections is a subset of another. The reason is the following: consider collections $\bfa$, $\mathbf{b}$, and $\mathbf{c}$, and suppose that $\mathbf{a}\subset \mathbf{b}$ (while $\bfa\not\subset \mathbf{c}$ and $\mathbf{c}\not\subset \bfa$). Then the information shared by all three collections is simply that shared by $\bfa$ and $\mathbf{c}$ since any information in $\bfa$ is automatically also contained in $\mathbf{b}$. In this way the information shared by multiple collections always reduces to the information associated with an antichain combination by removing all supersets. The measure $i_\cap
^{\sx}$ agrees with this result because the truth conditions of the statement $\mathcal{W}_{\bfa_1,\ldots,\bfa_m}$ are unaffected by superset removal.

Mathematically, the antichain-combinations form a lattice structure, i.e., there exists an ordering $\preceq$ of these antichain combinations such that for any pair of antichain combinations there is a unique infimum and supremum. In~\cite{williams2010nonnegative}, this lattice of antichain combinations is called the \emph{redundancy lattice} since it models inclusion of redundancies: redundant information terms associated with lower level antichains are included in redundancies associated with higher level antichains. Williams and Beer then introduced the PID terms implicitly via a M\"{o}bius Inversion over the lattice (more details in Appendix~\ref{apx:subsec:inf-misinf}). We can proceed in just the same way on a pointwise level and introduce the PPID terms via a M\"{o}bius-Inversion of $i_\cap^{\sx}$, i.e., via inverting the relationship
\begin{equation}
i_\cap^{\sx}(t:\alpha) = \sum_{\beta \preceq \alpha} \pi^{\sx}(t:\beta)
\end{equation}
where $\alpha$ and $\beta$ are antichain combinations. In this way each PPID term $\pi^{\sx}$ measures the information ``increment''  as we move up the lattice, i.e., the PPID term of a given node is that part of the corresponding shared information that is not already contained in any lower level shared information.

It should be mentioned at this point that the measure $i^{\sx}_{\cap}$ actually violates one of the W\&B axioms for shared information: it is not monotonically decreasing as more collections of source realizations are included. On first sight this appears to be a problem because one would expect, for instance, that the information shared by source realizations $s_1,s_2$ and $s_3$ should be \textit{smaller than or equal to} the information shared by $s_1$ and $s_2$. After all, the information shared by all three source realizations should be contained in the information shared by the first two. However, the violation of the monotonicity property has a natural interpretation in terms of informative and misinformative contributions to redundant information~\cite{finn2018pointwise}: whereas each of these components individually \textit{should} indeed satisfy the monotonicity axiom, this is not true of the total redundant information. Using the above example, the information shared by  $s_1,s_2,$ and $s_3$ can actually be larger than the information shared by $s_1$ and $s_2$ if the extra information in the latter shared information term (i.e., the information shared by $s_1$ and $s_2$ but \textit{not} by $s_3$) is misinformative.

As shown in \cite{finn2018probability} it is possible to uniquely decompose the pointwise mutual information into an informative and a misinformative component. Since $i^{\sx}_{\cap}$  is itself a pointwise mutual information the same decomposition can be applied in order to obtain an informative pointwise shared information $i^{\sx+}_{\cap}$ ~\eqref{eq:isx_decomposed-plus} and a misinformative pointwise shared information $i^{\sx-}_{\cap}$ ~\eqref{eq:isx_decomposed-minus}. We may then show that each of these components individually satisfies the W\&B axioms. The decomposition reads
\begin{subequations}\label{eq:isx_decomposed}
\small
\begin{align}
i^{\sx}_{\cap}(t:\bfa_1; \bfa_2;\ldots; \bfa_m) 	&=i_\cap^{\sx+}(t:\bfa_1; \bfa_2;\ldots; \bfa_m)\notag\\ 
                                                				&- i_\cap^{\sx-}(t:\bfa_1; \bfa_2;\ldots; \bfa_m) \label{eq:isx_decomposed-rlz},
\end{align}
\end{subequations}
\begin{subequations}
\small
    \begin{align}
        i_\cap^{\sx+}(t:\bfa_1; \bfa_2;\ldots; \bfa_m)  &:= \log_2\frac{1}{\mP( \mfa_1 \cup \mfa_2 \cup \ldots \cup \mfa_m) },\label{eq:isx_decomposed-plus}\\ 
        i_\cap^{\sx-}(t:\bfa_1; \bfa_2;\ldots; \bfa_m)  &:=\log_2\frac{\mP(\mft)}{\mP(\mft \cap (\mfa_1 \cup \mfa_2 \cup \ldots \cup \mfa_m)) }.\label{eq:isx_decomposed-minus}
    \end{align}
\end{subequations}

\noindent Here, the first term of~\eqref{eq:isx_decomposed-rlz} is considered to be the informative part as it is what can be inferred from the sources (recall that $\SE{a}_i$ are indices of collections of sources) and we refer to it by $i_\cap^{\sx+}$~\eqref{eq:isx_decomposed-plus}. The second term of~\eqref{eq:isx_decomposed-rlz} quantifies the (misinformative) relative loss of $p(t)$, the probability mass of the event $\mathfrak{t}$ (which actually happened) when excluding the mass of $\mathfrak{\bar a}_1\cap\mathfrak{\bar a}_2\cap\ldots\cap \mathfrak{\bar a}_n$ and we refer to it by $i_\cap^{\sx-}$~\eqref{eq:isx_decomposed-minus}.

Now, $i_\cap^{\sx\pm}$ should individually fulfill a pointwise version of the Williams and Beer axioms. These PPID axioms were described by Finn and Lizier~\cite{finn2018pointwise}.

\begin{axiom}[Symmetry]\label{ax:sym}
    $i^{+}_{\cap}$ and $i^{-}_{\cap}$ are invariant under any permutation $\sigma$ of collections of source events:
    \begin{align*}
        i^{+}_{\cap}(t:\bfa_1; \bfa_2; \ldots; \bfa_m) &= i^{+}_{\cap}(t:\sigma(\bfa_1); \sigma(\bfa_2); \ldots; \sigma(\bfa_m)),\\
        i^{-}_{\cap}(t:\bfa_1; \bfa_2; \ldots; \bfa_m) &= i^{-}_{\cap}(t:\sigma(\bfa_1); \sigma(\bfa_2); \ldots; \sigma(\bfa_m)).
    \end{align*}
\end{axiom}

\begin{axiom}[Monotonicity]\label{ax:mono}
    $i^{+}_{\cap}$ and $i^{-}_{\cap}$ decreases monotonically as more source events are included,
    \begin{align*}
        i^{+}_{\cap}(t:\bfa_1; \ldots; \bfa_m; \bfa_{m+1}) &\le i^{+}_{\cap}(t:\bfa_1;\ldots; \bfa_m),\\
        i^{-}_{\cap}(t:\bfa_1; \ldots; \bfa_m; \bfa_{m+1}) &\le i^{-}_{\cap}(t:\bfa_1; \ldots; \bfa_m),
    \end{align*}
    with equality if there exists $i\in[m]$ such that $\bfa_i\subseteq\bfa_{m+1}.$ 
\end{axiom}

\begin{axiom}[Self-redundancy]\label{ax:self}
    $i^{+}_{\cap}$ and $i^{-}_{\cap}$ for a single source event $\bfa$ equal $i^+$ and $i^-,$ respectively:
    \begin{align*}
        \small
        i^{+}_{\cap}(t:\bfa) &= h(\bfa) = i^+(t:\bfa),\\
        i^{-}_{\cap}(t:\bfa) &= h(\bfa\mid t)= i^-(t:\bfa).
    \end{align*}
    Therefore, $i_\cap(t:\bfa)=i(t:\bfa).$
\end{axiom}
Note that $i(t:\bfa) = i^+(t;\bfa ) - i^-(t;\bfa),$ which is the informative--misinformative decomposition of the pointwise mutual information derived by Finn and Lizier~\cite{finn2018probability}. The following theorem states that $i_\cap^{\sx\pm}$ result in a consistent PPID by showing that $i_\cap^{\sx+}$ and $i_\cap^{\sx-}$ individually fulfill the PPID axioms~\cite{finn2018pointwise} (the proof is deferred to appendix~\ref{apx:sec:proofs}.
\begin{theorem}\label{thm:axioms}
    $i^ {\sx+}_{\cap}$ and $i^ {\sx-}_{\cap}$ satisfy Axioms~\ref{ax:sym}, \ref{ax:mono}, and~\ref{ax:self}.
\end{theorem}
\noindent In this way the violation of monotonicity of the total shared information $i^{\sx}_{\cap}$ can be completely explained in terms of misinformative contributions. In fact, there is a another form of monotonicity that should hold as well: monotonicity over the redundancy lattice. As noted above the redundancy lattice models inclusion of redundancies. So we would expect lower level redundancies to be smaller than higher level redundancies. Again this form of monotonicity does not hold for $i^{\sx}_{\cap}$ itself but for its informative and misinformative components as expressed in the following theorem:
\begin{theorem}\label{thm:mono-antichain}
        $i^{\sx\pm}_{\cap}$ increase monotonically on the redundancy lattice.
\end{theorem}
There is another apparent problem that can be addressed using the separation into informative and misinformative components, namely, the fact that both $i^{\sx}_{\cap}$ as well as $\pi^{\sx}$ can be negative. This can be interpreted in terms of misinformation as well. To this end we define misinformative and informative PPID terms $\pi_\pm^{\sx}$ via M\"{o}bius Inversions of $i^{\sx\pm}_{\cap}$. These informative and misinformative components of the PPID terms can be obtained recursively from $i^{\sx\pm}_{\cap}$ (see appendix~\ref{apx:sec:proofs}). They stand in the relation $\pi^{\sx}= \pi_+^{\sx}-\pi_-^{\sx}$ to the PPID terms. Now, even though $\pi^{\sx}$ may be negative, its components $\pi_+^{\sx}$ and $\pi_-^{\sx}$ are non-negative.
\begin{theorem}\label{thm:non-neg}
    The atoms $\pi_\pm^{\sx}$ are non-negative.
\end{theorem}
\noindent In appendix~\ref{apx:sec:proofs}, we will provide the necessary tools to prove the above theorems, in particular, theorem~\ref{thm:non-neg}. To sum up, this section shows that $i_\cap^{\sx}$ results in a consistent and meaningful PPID. The apparent problems of violating monotonicity and non-negativity can be resolved by separating misinformative and informative components and showing that these components do satisfy the desired properties (for more discussion on the idea of misinformation within local Shannon information theory see Discussion). 

This concludes our discussion of the PPID induced by the $i_\cap^{\sx}$. The global, variable-level PID can be obtained by simply averaging the local quantities over all possible realizations of the source and target random variables. For a complete worked example of the XOR probability distribution see Figure \ref{fig:XOR_worked}, subfigure H in particular. In the next section we establish the differentiability of $i_\cap^{\sx}$ and $\pi^{\sx}$, an important advantage of these measures compared to other approaches.

\subsection[]{Differentiability of $i_\cap^{\sx}$ and $\pi_{\pm}^{\sx}$}
    We will discuss the differentiability of the PPID obtained by $i_\cap^{\sx}$.  This is a desirable property~\cite{wibral2017partialneuralgoal} that is proven to be lacking in some measures~\cite{makkeh2017bivariate,makkeh2018optimizingbroja,ince2017measuring} or evidently lacking for other measures since their definitions are based on the maximum or (minimum) of multiple information quantities.
    
    Let $\mathscr{A}([n])$ be the redundancy lattice (see section~\ref{apx:sec:proofs}), $(T, S_1, \dots, S_n)$ be discrete and finite random variables, and let us represent their joint probability distribution as a vector in $[0,1]^{|\mathcal A_T|\times |\mathcal A_{S_1}|\times\dots\times|\mathcal A_{S_n}|}.$ Thus, the set of all joint probability distributions of  $(T, S_1, \dots, S_n)$ forms a simplex that we denote by $\Delta_P.$  Note that $i_\cap^{\sx}$ and $\pi_\pm^{\sx}$ are functions of the probability distributions of $(T, S_1, \dots, S_n)$ and so they can be differentiable w.r.t. the probability distributions. Formally, for a given $(T, S_1, \dots, S_n)$, we show that $i_\cap^{\sx}$ and $\pi_\pm^{\sx}$ are differentiable over the interior of $\Delta_P.$ 
    
    Since $\log_2$ is continuously differentiable over the open domain $\RR_+$, then using definitions~\eqref{eq:isx_decomposed-plus} and~\eqref{eq:isx_decomposed-minus}, $i_\cap^{\sx+}$ and $i_\cap^{\sx-}$ are both continuously differentiable over the interior of $\Delta_p.$ Now, for $\alpha\in\mathscr A([n])$, using theorem~\ref{thm:moeb-inv} and proposition~\ref{prop:the-mapping}
    \begin{equation}\label{eq:diff_atom}
        \small
        \pi_+^{\sx}(t:\alpha) = \sum_{\substack{\gamma\in\mathscr P(\alpha^-\backslash\{\gamma_1\})}}(-1)^{|\gamma|}\log_2\bigg(\frac{p(\gamma) + d_1}{p(\gamma)}\bigg),
    \end{equation}
    where $\alpha^-= \{\gamma_1,\gamma_2,\dots,\gamma_k\}$ are the children of $\alpha$ ordered increasingly w.r.t.~their probability mass and  $\alpha^-:=\{\beta\in\mathscr A([n])\mid \beta\prec\alpha, \beta\preceq\gamma\prec\alpha\Rightarrow\beta=\gamma\}.$ Hence, $\pi_+^{\sx}$ is continuously differentiable over the interior of $\Delta_P$ since the function $\nicefrac{x+d_1}{x}$ and its inverse are continuously differentiable over the open domain $\RR_+.$ Similarly, $\pi_-^{\sx}$ is continuously differentiable over the interior $\Delta_P.$

\begin{figure}
    \centering
    \includegraphics[width=0.42\textwidth]{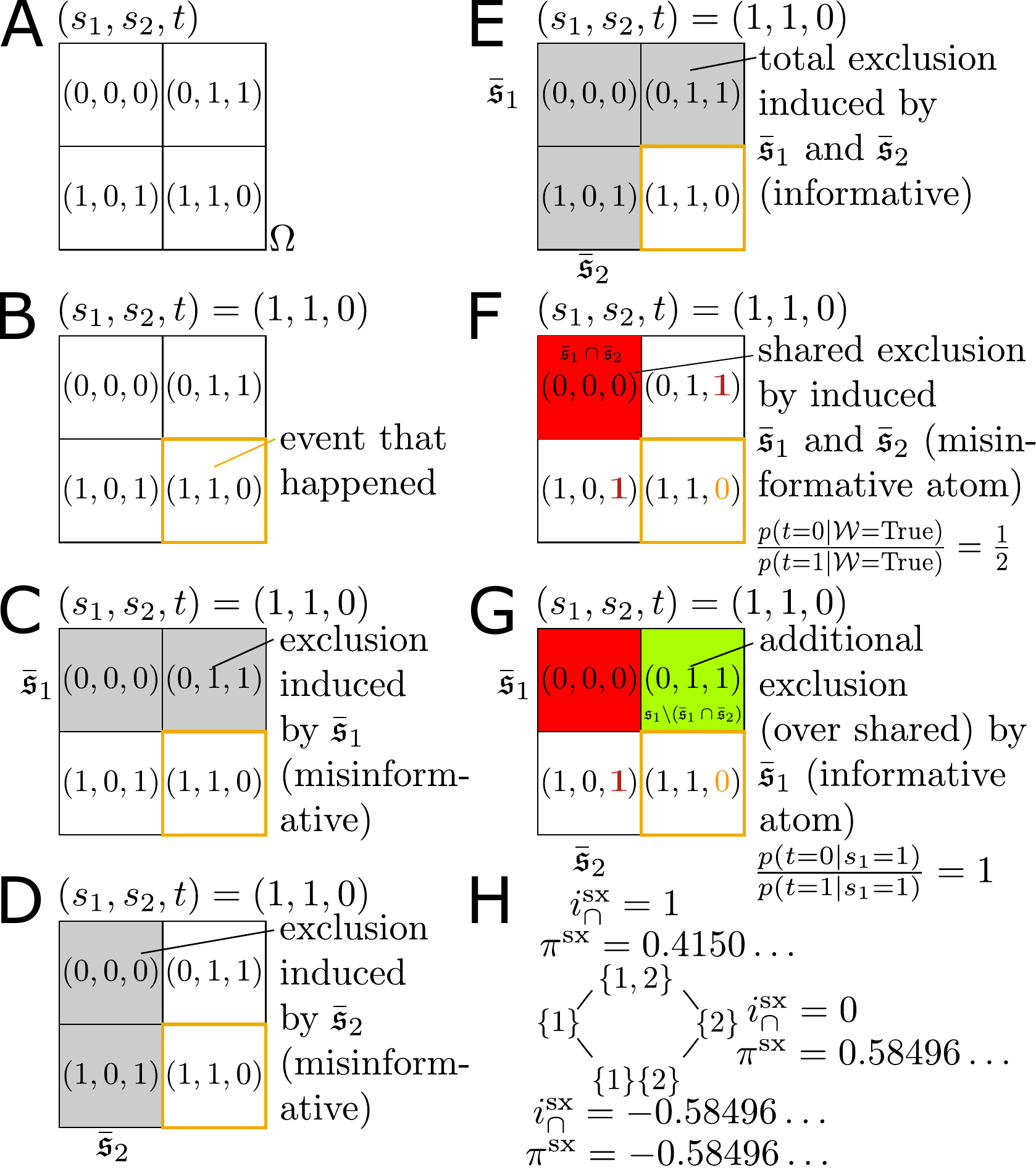}
    \caption{\label{fig:XOR_worked}\textbf{Worked example of $i_\cap^{\sx}$ for the classical \textsc{Xor}.} Let $T=XOR(S_1,S_2)$ and $S_1,S_2\in\{0,1\}$ be independent uniformly distributed and consider the realization $(s_1,s_2,t)=(1,1,0)$. (A-B) The sample space $\Omega$ and the realized event (gold (gray) frame). (C) The exclusion of events induced by learning that $S_1=1$, i.e. $\bmfs_1=\{0\}$ (gray). (D) Same for $\bmfs_2=\{0\}$. (E) The union of exclusions fully determines the event $(1,1,0)$ and yields 1 bit of $i(t=0: s_1=1, s_2=1)$. (F) The shared exclusions by $\bmfs_1=\{0\}$ and $\bmfs_2=\{0\}$, i.e., $\bmfs_1\cap\bmfs_2$ exclude only $(0,0,0)$. This is a misinformative exclusion, as it raises the probability of events that did not happen ($t=1$) relative to those that did happen ($t=0$) compared to the case of complete ignorance. (G) Learning about one full variable, i.e., obtaining the statement that $\bmfs_1=\{0\}$ adds additional probability mass to the exclusion (green (light gray)). The shared exclusion (red (dark gray)) and the additional unique exclusion (green (light gray)) induced by $s_1$ create an exclusion that is uninformative, i.e., the probabilities for $t=0$ and $t=1$ remain unchanged by learning $s_1=1$. At the level of the $\pi^{\sx}$ atoms, the shared and the unique information atom cancel each other. (H) Lattice with $i_\cap^{\sx}$ and $\pi^{\sx}$ terms for this realization. Other realizations are equivalent by the symmetry of XOR, thus, the averages yield the same numbers. Note that the necessity to cancel the negative shared information twice to obtain both $i(t=0:s_1=1)=0$ and $i(t=0:s_2=1)=0$, results in a synergy $<1$ bit. Also note that while adding the shared exclusion from (F) and the unique exclusions for $s_1$ and $s_2$ results in the full exclusion from (E), information atoms add differently due to the nonlinear transformation of excluded probability mass into information via $-\log_2 p(\cdot)$ -- compare (H).}
\end{figure}

\section{\label{sec:discussion}Discussion}
In this section, we first present further properties of $i^{\sx}_\cap.$ Then, we provide an operational interpretation of $i^{\sx}_\cap,$ and suggest an approach to compare this operational interpretation with that of other measures. Following this, we give the intuition behind the ``intrinsic dependence'' of PID atoms for joint source-target distributions where the number of these atoms is larger than these distributions' alphabet size. Finally, we provide two applications where $i^{\sx}_\cap$ is particularly well suited and discuss the computational complexity of $i^{\sx}_\cap.$ 

\subsection{Direct consequences of $i^{\sx}_{\cap}$ being a local mutual information}\label{sec:direct_consequences}
The fact that $i^{\sx}_{\cap}$ has the form of a regular local mutual information has several interesting consequences.

\paragraph{Implied entropy decomposition} Since the local entropy of a realization of a set of variables can be written as a self-mutual information our decomposition also directly implies an entropy decomposition that inherits the properties of the lattices described in section \ref{sec:lattice}. We start by the local entropy $h(\bfa_1,\ldots,\bfa_m)$ of a set of collections of realizations of variables $S_i=s_i$. Note that these collections have to be considered jointly, hence the comma~\footnote{If the collections where considered in an OR relation, there would be no random variable on which the average entropy is defined (see discussion of the local indicator variable $w_{\bfa_1,\ldots,\bfa_m}$)}. Thus, we can equally well write the entropy that is to be decomposed as $h(\{s_i\mid i \in \bigcup \bfa_j\})$. Thus, we can consider the $s_i$ together as a joint random variable whose entropy is to be decomposed. This can be done by realizing first $h(\{s_i\mid i \in \bigcup \bfa_j\})=i(\{s_i\mid i \in \bigcup \bfa_j\}:\{s_i\mid i \in \bigcup \bfa_j\})$, and then applying our PID formalism. In this decomposition then terms of the form $i_\cap^{\sx} (\{s_i\mid i \in \bigcup \bfa_j\}:\bfa_1;\ldots;\bfa_m)=:h_\cap^{\sx}(\bfa_1;\ldots;\bfa_m)$ appear. In other words, on the target side of the arguments of $i_\cap^{\sx}$ we will always find the joint random variable, whereas the collections appear as usual on the source side.

\paragraph{\label{sec:TargetChainRule}Target chain rule and average measures} Another consequence is that $i_\cap^{\sx}$ satisfies a target chain rule for a composite target variable $T=\{t_1,t_2\}$:

\[
\small
\begin{split}
i^{\sx}_{\cap}(t_1,t_2: \bfa_1; \bfa_2; \ldots; \bfa_m) &= i^{\sx}_{\cap}(t_1: \bfa_1; \bfa_2; \ldots; \bfa_m)\\
                                                        &+i^{\sx}_{\cap}(t_2: \bfa_1; \bfa_2; \ldots; \bfa_m\mid t_1),
\end{split}
\]
\noindent where the second term is $\log_2\frac{\mP(\mft_2\mid\mft_1) - \mP(\mft_2, \bmfa_1, \bmfa_2, \ldots, \bmfa_m\mid\mft_1)}{1 - \mP(\bmfa_1, \bmfa_2, \ldots, \bmfa_m\mid \mft_1)} - \log_2 \mP(\mft_2\mid \mft_1).$ Moreover, by linearity of the averaging a corresponding target chain rule is satisfied for the average shared information,  $I_\cap^{\sx}$, defined by.

{\small
\begin{align}\label{eq:isx_average}
I^{\sx}_{\cap}(T: \mathbf{A}_1;\ldots; \mathbf{A}_m)   &:=\sum_{t, s_1, \ldots, s_n} p(t, s_1, \ldots, s_n) i^{\sx}_\cap(t: \bfa_1; \ldots; \bfa_m)\notag\\
                &=\sum_{t, s_1, \ldots, s_n} p(t, \mathbf{s}_1, \ldots, \mathbf{s}_n) i(t:W_{\bfa_1,\ldots,\bfa_m}=1),
\end{align}
}
\noindent where probabilities related to the indicator variable $W_{\bfa_1,\ldots,\bfa_m}$ have to be recomputed for each possible combination of source and target realizations. Note that this indicator variable simply indicates the truth of the statement $\mathcal{W}_{\bfa_1,\ldots,\bfa_m}$ from section \ref{sec:Logic}. Also note that in Eq.~\eqref{eq:isx_average} the averaging still runs over all combinations of $t,s_1,\ldots,s_n$, and the weights are still given by $p(t,s_1,\ldots,s_n)$, not $p(t,W_{\bfa_1,\ldots,\bfa_m}=1)$. Having different variables in the averaging weights and the local mutual information terms makes the average shared information structurally different from a mutual information \footnote{As was to be expected from the difficulties encountered in the past trying to define measures of shared information.}. One consequence of this is that in principle the average $I^{\sx}_{\cap}$ can be negative. This also holds for the averages of the other information atoms on the lattice (see next section for the lattice structure).  Thus, the local shared information may be expressed as a local mutual information with an auxiliary variable constructed for that purpose, and multiple such variables have to be constructed for a definition of a global shared information.

\paragraph{Upper bounds.} First, we can assess the self-shared information of a collection of variables:
\begin{equation}\label{eq:self_shared}
    \small
    \begin{split}
        i^{\sx}_{\cap}(\bfa_1;\ldots; \bfa_m:\bfa_1;\ldots; \bfa_m) &:=i(W_{\bfa_1,\ldots,\bfa_m}=1:W_{\bfa_1,\ldots,\bfa_m}=1)\\
        &=h(W_{\bfa_1,\ldots,\bfa_m}=1)~,
    \end{split}
\end{equation}
\noindent where the notation $\bfa_1;\ldots; \bfa_m$ means the event defined by the complement of the intersection of exclusions induced by the $\bfa_i$, as before. This quantity is greater than or equal to zero and is the upper bound of shared information that the source variables can have about any realization $u$ of any target variable $U$, i.e.,
{\small
\[
    i^{\sx}_{\cap}(\bfa_1;\bfa_2;\ldots;\bfa_m:\bfa_1;\bfa_2;\ldots;\bfa_m) \geq i^{\sx}_{\cap}(u: \bfa_1;\bfa_2;\ldots;\bfa_m)
\]
}\noindent for any $u \in \mathcal{A}_U.$ This upper bound has conceptual links to maximum extractable shared information from \cite{rauh2017extractable}. Moreover, this upper bound may be nonzero even for independent sources, showing how the so-called mechanistic shared information arises.

\subsection[]{Operational interpretation of $i_\cap^{\sx}$}
Being a local mutual information, $i_\cap^{\sx}$ keeps all the operational interpretations of that measure. For example, in keeping with Woodward \cite{woodward1952information} it measures the information available in the statement $\mathcal{W}$ for inference about the value $t$ of the target. Specifically, a negative value of the local shared information indicates that an agent who is only in possession of the shared information is more likely to mispredict the outcome of the target (e.g., FIGs~\ref{fig:XOR_worked},~\ref{fig:4bit_parity_worked}) than without the shared information; a positive value means that the shared information makes the agent more likely to choose the correct outcome. The unsigned magnitude of the shared information informs us about how relatively certain the agent should be about their prediction.

What remains to be clarified then is the meaning of the average expression $I_\cap^{\sx}$. As detailed above the average is taken with respect to the probabilities of the realizations of the source variables and the target variable, not with respect to the dummy variables encoding the truth value of the respective statements $\mathcal{W}$ --- as an average mutual information would require. To understand the meaning of this particular average it is instructive to start by ruling out two false interpretations. Again, consider an agent who tries to predict the correct value of target $t.$ In order to do so, the agent utilizes a particular information channel.

For the first false interpretation, consider a channel that takes the realizations of sources and target and produces the statements $\mathcal{W}$ carrying the shared information. If the receiver of this channel used it multiple times in the case of a negative $I_\cap^{\sx}$, then this receiver would \emph{learn} that the shared information received is negative on average and could modify their judgment. This leads us to a second false interpretation: the average could be understood as an average over an \emph{ensemble} of agents, where each agent uses the above channel only once, thus avoiding the issue just described. Even in this scenario however there is a problem: if the agent knew that the information provided by $\mathcal{W}$ is shared by the true source realizations, then the agent could derive the truth of all sub-statements of $\mathcal{W}$. Accordingly, the agents would receive more than only the shared information.

In order to obtain the appropriate interpretation of shared information we have to consider a channel that masks the metainformation that all substatements of $\mathcal{W}$ are true, and also makes learning impossible. This is achieved by a channel that produces true statements $\mathcal{V}$ about the source variables which have the logical structure of $\mathcal{W}$, but do not always carry shared information. Consider the information shared by all sources. In this case the channel would randomly produce (true) statements of the form $ \mathcal{V}_{s_1,\ldots,s_n} = \big((S_1=s_1) \vee \ldots \vee  (S_n=s_n)\big)$ but where some of the substatements might be false. Then $\mathcal{V}$ does not always carry shared information (only in case all substatements happen to be true). The receiver knows the joint distribution of sources and target and performs inference on $t$ in a Bayes optimal way. Such a channel would provide non-negative average mutual information. However, for a channel of this kind the average taken to compute $I_\cap^{\sx}$, is only over those channel uses where $\mathcal{V}$ actually did encode shared information. In certain cases this average can be negative (see Table \ref{tab:V_channel}). 

As already alluded to above, the setting of our operational interpretation contrasts with that of other approaches to PID that take the perspective of multiple agents having \emph{full} access to individual source variables (or collections thereof), and that then design measures of unique and redundant information based on actions these agents can take or rewards  they obtain in decision- or game-theoretic settings based on their access to full source variables (e.g., in~\cite{bertschinger2014quantifying,finn2018pointwise,ince2017measuring}). While certainly useful in the scenarios invoked in~\cite{bertschinger2014quantifying,finn2018pointwise,ince2017measuring}, we feel that these operational interpretations may almost inevitably mix inference problems (i.e., information theory proper) with decision theory. Also, they typically bring with them the use of minimization or maximization operations to satisfy the competitive settings of decision or  game theory. This, in turn, renders it difficult to obtain a differentiable measure of local shared information.

In sum, we feel that the question of how to decompose the information provided by multiple source variables about a target variable may indeed not be a single question, but multiple questions in disguise. The most useful answer will therefore depend on the scenario where the question arose. Our answer seems to be useful in communication settings, and where quantitative statements about dependencies between variables are important (e.g., the field of statistical inference, where the PID enumerates all possible types of dependencies of the dependent (target) variable on the independent (source) variables).
\begingroup
    \squeezetable
    \begin{table}
        \caption{\label{tab:V_channel}\textbf{$\mathcal{V}$-channel for \textsc{Xor}}. \emph{Left:} probability masses for each realization. \emph{Middle:} Equiprobable $\mathcal{V}$-statements associated with each realization such that respective statement carrying shared information is listed first (marked by $\mathcal{W}$) \emph{Right:} predicted target inferred from $\mathcal{V}$ and where \checkmark refers to \emph{correct} predictions and \ding{55} refers to \emph{incorrect} ones. Using $\mathcal{V}$ a receiver obtains positive average mutual information, but the contribution of $\mathcal{W}$ statements is negative. \emph{Bottom:} the sign of $I^{\mathcal V}$, the average information provided by all $\mathcal {V}$-statements, and that of $I_\cap^{\sx}$.}
        \centering
        \begin{tabular}{c||cc|c||c||cc}
            \toprule
            \multicolumn{1}{c||}{}    &\multicolumn{3}{c||}{Realization}    &\multicolumn{1}{c||}{Channel Output} &\multicolumn{2}{c}{Inference}\\
            \cmidrule(r){1-1}\cmidrule(lr){2-4}\cmidrule(lr){5-7}
             $p$                &$s_1$  &$s_2$  &$t$    &$\mathcal V$-statement            & predicted $t$ &Correct?   \\
             \nicefrac{1}{4}    &0      &0      &0      & $(S_1=0)\vee(S_2 =0)~(\mathcal{W})$    &1                  &\ding{55}\\
                                &       &       &       & $(S_1=0)\vee(S_2 =1)$             &0                  &\checkmark\\
                                &       &       &       & $(S_1=1)\vee(S_2 =0)$             &0                  &\checkmark\\
             \cmidrule{1-7}
             \nicefrac{1}{4}    &0      &1      &1      & $(S_1=0)\vee(S_2 =1)~(\mathcal{W})$    &0                  &\ding{55}\\
                                &       &       &       & $(S_1=0)\vee(S_2 =0)$             &1                  &\checkmark\\
                                &       &       &       & $(S_1=1)\vee(S_2 =1)$             &1                  &\checkmark\\
             \cmidrule{1-7}
             \nicefrac{1}{4}    &1      &0      &1      & $(S_1=1)\vee(S_2 =0)~(\mathcal{W})$    &0                  &\ding{55}\\
                                &       &       &       & $(S_1=1)\vee(S_2 =1)$             &1                  &\checkmark\\
                                &       &       &       & $(S_1=0)\vee(S_2 =0)$             &1                  &\checkmark\\
             \cmidrule{1-7}
             \nicefrac{1}{4}    &1      &1      &0      & $(S_1=1)\vee(S_2 =1)~(\mathcal{W})$    &1                  &\ding{55}\\
                                &       &       &       & $(S_1=1)\vee(S_2 =0)$             &0                  &\checkmark\\
                                &       &       &       & $(S_1=0)\vee(S_2 =1)$             &0                  &\checkmark\\
             \bottomrule
        \end{tabular}
        \vskip 5pt{\tiny $I^{\mathcal V}(T: S_1; S_2) > 0~\text{(4 \ding{55} and 8 \checkmark)}\quad\text{and}\quad I_\cap^{\sx}(T: S_1; S_2) < 0~\text{(4 \ding{55} and 0 \checkmark})$}
    \end{table}
\endgroup
\subsection{Evaluation of $I^{\sx}_\cap$ on $P$ and on optimization distributions obtained in other frameworks.} Since our approach to PID relies only on the original joint distribution $P$ it can be applied to other PID frameworks where distributions $Q(P)$ are derived from the original $P$ of the problem -- e.g., via optimization procedures, as it is done for example in~\cite{bertschinger2014quantifying,ince2017measuring}. This yields some additional insights into the operational interpretation of our approach compared to others, by highlighting how the optimization from $P$ to $Q(P)$ shifts information between PID atoms in our framework.

\subsection{Number of PID atoms vs alphabet size of the joint distribution}
The number of lattice nodes rises very rapidly with increasing numbers of sources. Thus, the number of lattice nodes may outgrow the joint symbol count of the random variables, i.e., the number of entries in the joint probability distribution. One may ask, therefore, about the independence of the atoms on the lattice in those cases (remember that the atoms were introduced in order to have the ``independent'' information contributions of respective variable configurations at the lattice nodes). As shown in Fig.~\ref{fig:Lattice-map} and~\ref{fig:d1} our framework reveals multiple additional constraints at the level of exclusions via the family of mappings from Proposition~\ref{prop:the-mapping}. This explains mechanistically why not all atoms are independent in cases where the number of atoms is larger than the number of symbols in the joint distribution.

\subsection{Key applications}
Due to the fact that PID solves a basic information-theoretic problem, its applications seem to cover almost all fields where information theory can be applied. Here, we focus on two applications for which our measure is suited particularly well: the first application requires localizability and differentiability; the second application does not require differentiability, but requires at least continuity of the measure on the space of the underlying probability distributions.

\subsubsection{Learning neural goal functions}
In~\cite{wibral2017partialneuralgoal} we argued that information theory, and in particular the PID framework, lends itself to unify various neural goal functions, e.g., infomax and others. We also showed how to apply this to learning in neural networks via the coherent infomax framework of Kay and Phillips~\cite{kay2011coherent}. Yet, this framework was restricted to goal functions expressible using combinations, albeit complex ones, of terms from classic information theory, due to the lack of a differentiable PID measure. Goal functions that were only expressible using PID proper could not be learned in the Kay and Phillips framework, and in those cases PID would only serve to assess the approximation loss. 

Our new measure removes this obstacle and neural networks or even individual neurons can now be devised to learn pure PID goal functions. A possible key application is in hierarchical neural networks with a hierarchy of modules, where each module contains two populations of neurons. These two populations represent supra- and infragranular neurons and coarsely mimic their different functional roles. One population represents  so-called layer 5 pyramidal cells. It serves to send the shared information between their bottom-up (e.g., sensory) inputs  and their top-down (contextual) inputs downwards in the hierarchy; the other population represents layer 3 pyramidal cells and sends the synergy  between the bottom-up inputs and the top-down  inputs upwards in the hierarchy. For the first population the extraction of shared information between higher and lower levels in the hierarchy can be roughly equated to learning an internal model, while for the second population the extraction of synergy is akin to computing a generalized error (see \cite{bastos2012canonical,larkum2013cellular} and references therein for the neuroanatomic background of this idea). Thus, a hierarchical network of this kind can perform an elementary type of predictive coding. The full details of this application scenario are the topic of another study, however.

\subsubsection{Information modification in distributed computation in complex systems}
If one desires to frame distributed computation in complex systems in terms of the elementary operations on information performed by a Turing machine, i.e., the storage, transfer, and modification of information, information-theoretic measures for each of these component operations are required. For storage and transfer well established measures are available, i.e., the active information storage~\cite{lizier2012localstorage} and the transfer entropy~\cite{schreiber2000measuring, wibral2013measuring,lizier2008localtransfer}. For modification, in contrast, no established measures exist, yet an appropriate measure of synergistic mutual information from a partial information decomposition has been proposed as a candidate measure of information modification~\cite{lizier2013towardsinfomodif}. An appropriate measure in this context has to be localizable (i.e., it must be possible to evaluate the measure for a single event) in order to serve as an analysis of computation locally in space and time, and it has to be continuous in terms of the underlying probability distribution. Both of these conditions were already met for the PPID measure of Finn and Lizier~\cite{finn2018pointwise}; our novel measure here adds the possibility to differentiate the measure on the interior of the probability simplex, which makes it even more like a classic information measure. This is important to determine the input distribution that maximizes synergy in a system, i.e., the input distribution that reveals the information modification capacity of the computational mechanism in a system as suggested in~\cite{wibral2017modification}.

\subsection[]{Computational complexity of the PID using $i^{\sx}_\cap$}
Real-world applications of PID will not necessarily be confined to the standard two-input variable case -- hence the importance of the organization scheme for higher order terms that are provided by the lattice structure. For such real-world problems the computational complexity of the computation of each atom on the lattice becomes important -- not least because of the potentially large number of atoms (see below). This holds in particular when additional nonparametric statistical tests of PID measures obtained from data require many recomputations of the measures. We, therefore, discuss the computational complexity of our approach.

For each realization $s=(s_1,\ldots,s_n)$ and $t$, our PPID is obtained by computing the atoms $\pi^{\sx}_\pm(t:\alpha)$ for each $\alpha\in\mathscr A([n]).$ In Appendix~\ref{apx:sec:proofs}, we show that any $\pi^{\sx}_\pm(t:\alpha)$ is evaluated as follows:
\[\small
    \pi_\pm^{\sx}(t:\alpha) =  i^{\sx\pm}_{\cap}(t:\alpha) - \sum_{\beta\prec\alpha}\pi_\pm^{\sx}(t:\beta)\quad\forall~\alpha,\beta\in\mathscr A([n]),
\]
where computing any $i^{\sx}_\cap(t:\alpha)$ is linear in the size of $\mathcal A_{T,S}$, the alphabet of the joint random variable $(T,S_1,\ldots,S_n).$ Moreover, using $i^{\sx}_\cap$ as a redundancy measure, the closed form of $\pi_\pm^{\sx}$ derived in~\eqref{eq:diff_atom} shows that the computation of our PID is trivially parallelizable over atoms and realizations, which is crucial for larger number of sources. The importance of parallelization is due to the rapid growth of PID terms $M$ when the number of sources gets larger for any PID lattice-based measure. This $M$ grows super exponentially as the $n$-th Dedekind number $d(n)-2$. At present even enumerating $M$ is practically intractable beyond $n>8$.

\section{\label{sec:examples}Examples}
    In this section, we present the PID provided by our $i_\cap^{\sx}$ measure for some exemplary probability distributions. Most of the distributions are chosen from Finn and Lizier~\cite{finn2018pointwise} and previous examples in the PID literature. The code for computing $\pi^{\sx}$ is available on the IDTxl toolbox~\url{http://github.com/pwollstadt/IDTxl}~\cite{wollstadt2018idtxl}.
    
    \subsection[]{Probability distribution \textsc{PwUnq}}
    We start by the pointwise unique distribution (\textsc{PwUnq}) introduced by Finn and Lizier~\cite{finn2018pointwise}. This distribution is constructed such that for each realization, only one of the sources holds complete information about the target while the other holds no information. The aim was to structure a distribution where at no point (realization) the two sources give the same information about the target. Hence, Finn and Lizier argue that, for such distribution, there should be no shared information.  Also, this distribution highlights the need for a pointwise analysis of the PID problem.
    
    Since in all of the realizations, the shared exclusion does not alter the likelihood of any of the target events compared to the case of total ignorance, $i_\cap^{\sx}$ will indeed give zero redundant information. Thus, the PID terms resulting from $i_\cap^{\sx}$ are the same as the those resultant from $r_{\min}$~\cite{finn2018pointwise} and $I_{\mathrm{ccs}}$~\cite{ince2017measuring} measures (see table~\ref{tab:pwunq}).
    \begingroup
    \squeezetable
    \begin{table}
        \caption{\label{tab:pwunq}\textbf{\textsc{PwUnq} Example}. \emph{Left:} probability mass diagrams for each realization. \emph{Right:} the pointwise partial information decomposition for the informative and misinformative. \emph{Bottom:} the average partial information decomposition.}
        \centering
        \begin{tabular}{c||cc|c||cccc|cccc}
            \toprule
            \multicolumn{1}{c||}{}    &\multicolumn{3}{c||}{Realization}    &\multicolumn{4}{c|}{$\pi_+^{\sx}$} &\multicolumn{4}{c}{$\pi_-^{\sx}$}\\
            \cmidrule(r){1-1}\cmidrule(lr){2-4}\cmidrule(lr){5-8}\cmidrule(l){9-12}
             $p$                &$s_1$  &$s_2$  &$t$    &$\{1\}\{2\}$   &$\{1\}$    &$\{2\}$    &$\{1,2\}$   &$\{1\}\{2\}$ &$\{1\}$   &$\{2\}$    &$\{1,2\}$\\
             \nicefrac{1}{4}    &0         &1          &1       &1              &0          &1          &0          &1 
             &0         &0          &0\\
             \nicefrac{1}{4}    &1         &0          &1       &1              &1          &0          &0          &1 
             &0         &0          &0\\
             \nicefrac{1}{4}    &0         &2          &2       &1              &0          &1          &0          &1 
             &0         &0          &0\\
             \nicefrac{1}{4}    &2         &0          &2       &1              &1          &0          &0          &1 
             &0         &0          &0\\
             \cmidrule{1-12}
             \multicolumn{4}{c||}{Average Values}                  &1              &\nicefrac{1}{2} &\nicefrac{1}{2} &0&1
             &0         &0          &0\\
             \bottomrule
        \end{tabular}
        \vskip 5pt{\tiny $\Pi^{\sx}(T:\{1\}\{2\})=0 \hspace{0.1cm} \Pi^{\sx}(T:\{1\})=\nicefrac{1}{2} \hspace{0.1cm} \Pi^{\sx}(T:\{2\})=\nicefrac{1}{2} \hspace{0.1cm} \Pi^{\sx}(T:\{1,2\})=0$}
    \end{table}
    \endgroup
    
    Recall Assumption $(*)$ of Bertschinger et al.~\cite{bertschinger2014quantifying} which states that the unique and shared information should only depend on the marginal distributions $P(S_1, T)$ and $P(S_2, T).$ Finn and Lizier~\cite{finn2018pointwise} showed that all measures which satisfy Assumption $(*)$ result in no unique information, i.e., nonzero redundant information whenever $P(S_1, T)$ is isomorphic to $P(S_2, T).$ The \textsc{PwUnq} distribution falls into this category for which $I_{\min}$~\cite{williams2010nonnegative}, $I_{\red}$~\cite{harder2013bivariate}, $\widetilde{\UI}$~\cite{bertschinger2014quantifying}, and $\mathcal S_{\vk}$~\cite{griffith2014quantifying} do not register unique information of $S_1$ and $S_2$. This is due to Assumption $(*)$ not taking into consideration the pointwise nature of information. Specifically, a measure that satisfies Assumption $(*)$ is agnostic to the fact that at each realization $\{T=j\}$ is uniquely determined by $S_1$ or $S_2$ but never both. On the contrary such a measure registers this as a mixture of shared and synergistic contribution since neither $S_1$ nor $S_2$ can fully determine $\{T=j\}$ on its own but shared they partly determine $\{T=j\}$.
    
    \subsection{Probability distribution XOR}
Using our formulation of  $i^{\sx}_{\cap}$ results in negative local shared information for the classic XOR example. To see this, assume that $S_1$ and $S_2$ are independent, uniformly distributed random bits, and $T=\mathrm{XOR}(S_1,S_2)$, and consider the realization $(s_1,s_2,t)=(1,1,0)$. From Eq.~\eqref{eq:isx} we get
{\small
\[
i^{\sx}_{\cap}(t=0:s_1=1;s_2=1) =\log_2\frac{1/2-1/4}{1-1/4} + \log_2\frac{1}{1/2}<0.
\]
}
We argue that this result reflects that an agent receiving the shared information is misinformed (see, e.g., \cite{finn2018probability} for the concept of misinformation) about $t$. To understand the source of this misinformation, consider that the agent is only provided with the shared information, i.e., the agent knows only that $\mathcal{W}_{s_1,\ldots,s_n}$ is true. This means the agent is being told the following: ``One of the two sources has outcome 1, and we do not know which one.'' This will let the agent predict that the joint realization is one out of three realizations with equal probability:  $(1, 1,0)$,  $(0, 1, 1),$ or $(1, 0, 1)$ (see FIG~\ref{fig:XOR_worked}). Of these three realizations, only one points to the correct target realization $t=0$, while the other two point to the ``wrong'' $t=1$ leading to odds of 1:2 --- whereas $t=0$ and $t=1$ were equally probable before the agent received the shared information from the sources. As a consequence, the local shared information becomes negative~\footnote{Due to $i(t:s_j)=0$ for $j=1,2$ in the XOR example, this negative shared information is then compensated by positive unique information -- however this happens twice, i.e. once for each marginal local mutual information. As a consequence, the synergy is reduced from 1 bit to 1 minus once this unique information. This may seem counter-intuitive when still thinking about the PID atoms as areas, in the sense of ``How come if we subtract two mutual information of zero bit from the joint mutual information of 1 bit, that we do not get 1 bit as a result?''. The key insight is that the two local mutual information terms of zero bit have a negative ``overlap'' with each other, making their sum positive. We simply see here again that the interpretation of PID atoms as (semi-positive) areas has to be given up in the pointwise framework, due to the fact that already the regular local mutual information can be negative.}. Finally, the XOR gate demonstrates an example of negative shared information; we note that in general unique (e.g., table~\ref{tab:rnderr}) and synergistic information can as well be negative.

    \subsection{Probability distribution \textsc{RndErr}}
        Recall \textsc{Rnd}, the redundant probability distribution, where both sources are fully informative about the target and exhibit the same information. More precisely, the \emph{redundant realizations}, $s_1=s_2=t=0$ and $s_1=s_2=t=1,$ are the only two realizations that occur equally likely. Derived from \textsc{Rnd}, the \textsc{RndErr} is a noisy redundant distribution of two sources where one source occasionally misinforms about the target while the other remains fully informative about the target. Moreover, if $S_2$ is the source that occasionally misinforms about the target, then the~\emph{faulty realizations}, namely, $s_2\neq s_1=t=0$ and $s_2\neq s_1=t=1,$ are equally likely, but  less likely than the redundant ones. We stick to the probability masses given in~\cite{finn2018pointwise} for the redundant realizations $\nicefrac{3}{8}$ and for the faulty realizations $\nicefrac{1}{8}$ and speculate that $S_2$ will hold misinformative (negative) unique information about $T.$
        
        For this distribution, our measure results in the following PID: misinformative unique information by $S_2$, informative unique information by $S_1$, informative shared information, and informative synergistic information that balances the misinformation of $S_2$ (see table~\ref{tab:rnderr}).
        \begingroup
        \squeezetable
        \begin{table}
        \caption{\label{tab:rnderr}\textbf{\textsc{RndErr} Example.} \emph{Left:} probability mass diagrams for each realization. \emph{Right:} the pointwise partial information decomposition for the informative and misinformative is evaluated. \emph{Bottom:} the average partial information decomposition. We set $a=\log_2(\nicefrac{8}{5}), b=\log_2(\nicefrac{8}{7}), c=\log_2(\nicefrac{5}{4}), d=\log_2(\nicefrac{7}{4}), e=\log_2(\nicefrac{16}{15}), f=\log_2(\nicefrac{16}{17}),$ and $g=\log_2(\nicefrac{4}{3}).$}
        \centering
        \begin{tabular}{c||cc|c||cccc|cccc}
            \toprule
            \multicolumn{1}{c||}{}    &\multicolumn{3}{c||}{Realization}    &\multicolumn{4}{c|}{$\pi_+^{\sx}$} &\multicolumn{4}{c}{$\pi_-^{\sx}$}\\
            \cmidrule(r){1-1}\cmidrule(lr){2-4}\cmidrule(lr){5-8}\cmidrule(l){9-12}
             $p$                &$s_1$  &$s_2$  &$t$    &$\{1\}\{2\}$               &$\{1\}$
             &$\{2\}$                           &$\{1,2\}$                       &$\{1\}\{2\}$               &$\{1\}$   
             &$\{2\}$                           &$\{1,2\}$\\
             \nicefrac{3}{8}    &0          &0          &0      &$a$  &$c$          &$c$         &$e$    &0                          &0 &$g$           &0\\
             \nicefrac{3}{8}    &1              &1          &1      &$a$  &$c$          &$c$         &$e$  &0                          &0 &$g$           &0\\
             \nicefrac{1}{8}    &0              &1          &0      &$b$  &$d$          &$d$         &$f$   &0                          &0 
             &2           &0\\
             \nicefrac{1}{8}    &1              &0          &1      &$b$  &$d$          &$d$         &$f$   &0                          &0 
             &2           &0\\
             \cmidrule{1-12}
             \multicolumn{4}{c||}{Average Values}                  &0.557                    &0.443    
             &0.443                             &0.367                          &0                          &0         
             &0.811     &0\\
             \bottomrule
        \end{tabular}
        \vskip 5pt
        {\tiny $\Pi^{\sx}(T:\{1\}\{2\})=0.557 \hspace{0.1cm} \Pi^{\sx}(T:\{1\})=0.443 \hspace{0.1cm} \Pi^{\sx}(T:\{2\})=-0.367 \hspace{0.1cm} \Pi^{\sx}(T:\{1,2\})=0.367$}
        \end{table}
        \endgroup
    \subsection{Probability distribution \textsc{XorDuplicate}}
    In this distribution, we extend the \textsc{Xor} distribution by adding a third source $S_3$ such that (i) $S_3$ is a copy of any of the two original sources and (ii) $S_3$ does not have an additional effect on the target, e.g., if $S_3$ is a copy of $S_1$ then $T :=\textsc{Xor}(S_1, S_2)= \textsc{Xor}(S_2, S_3).$ Let $S_1$ and $S_2$ be two independent, uniformly distributed random bits, $S_3$ be a copy of $S_1,$ and $T=\textsc{Xor}(S_1, S_2).$ This distribution $(S_1, S_2, S_3, T)$ is called \textsc{XorDuplicate} where the only nonzero realizations are $(0,0,0,0), (0,1,0,1), (1,0,1,1), (1,1,1,0).$ 
    
    The key point is that the target $T$ in the classical $\textsc{Xor}$ is specified only by $(S_1, S_2)$, whereas in \textsc{XorDuplicate} the target is equally specified by the coalitions $(S_1,S_2)$ and $(S_2,S_3)$. This means that the synergy $\Pi^{\sx}(T:\{1,2\})$ in \textsc{Xor} should be captured by the term $\Pi^{\sx}(T:\{1,2\}\{2,3\})$ in \textsc{XorDuplicate}. 
    
    The \textsc{XorDuplicate} distribution was suggested by Griffith et al.~\cite{griffith2014quantifying}. The authors speculated that their definition of synergy $\mathcal S_{\vk}$ must be invariant to duplicates for this distribution, $\Pi^{\sx}(T:\{1,2\}\{2,3\}) = \Pi^{\sx}(T:\{1,2\}),$ since the mutual information is invariant to duplicates, $I(T: S_1, S_2, S_3) = I(T:, S_1, S_2).$ Also, they proved that $\mathcal S_{\vk}$ is invariant to duplicates in general~\cite{griffith2014quantifying}. 
    
    For the shared exclusion measure $i_\cap^{\sx}$, it is evident that the invariant property will hold since the shared information is indeed a mutual information and it is easy to see that $i_\cap^{\sx}(t: s_1; s_2; s_3) = i_\cap^{\sx}(t: s_1; s_2).$ In fact, we show below that all the PID terms are invariant to the duplication. That is, the unique information of $S_2$ is invariant and captured by $\Pi^{\sx}(T:\{2\})$. Also, the unique information of $S_1$ is invariant but is captured by the atom $\Pi^{\sx}(T:\{1\}\{3\})$ since it is shared information by $S_1$ and $S_3$ as $S_3$ is a copy of $S_1$. Finally, the synergistic information is invariant, however, it is captured by $\Pi^{\sx}(T:\{1,2\}\{2,3\})$ since the coalitions $(S_1,S_2)$ and $(S_2,S_3)$ can equally specify the target. These claims are shown below by replacing $s_3$ by $s_1$ and applying the monotonicity axiom~\ref{ax:mono} on $i_\cap^{\sx+}$ and $i_\cap^{\sx-}.$ Note that due to symmetry all the realizations have equal PID terms and the difference between the informative and misinformative is computed implicitly. 
    
    For any $(t,s_1,s_2,s_3)$ with nonzero probability mass, we have
    {\small
    \begin{align*}
        &i_\cap^{\sx}(t: s_1; s_2; s_3)= i_\cap^{\sx}(t: s_1; s_2)= i_\cap^{\sx}(t: s_2; s_3)= -0.5849\\
        &i_\cap^{\sx}(t: s_1; s_3)= i_\cap^{\sx}(t: s_1)  = i_\cap^{\sx}(t:s_3)= 0
    \end{align*}
    }
    implying that
    {\small
    \begin{align*}
        &\pi^{\sx}(t:\{1\}\{2\})= \pi^{\sx}(t:\{2\}\{3\}) = 0\\
        &\pi^{\sx}(t:\{1\}\{3\})= -\pi^{\sx}(t:\{1\}\{2\}\{3\}) = 0.5849.
    \end{align*}
    }
    But, $i_\cap^{\sx}(t: s_2; s_1, s_3)=i_\cap^{\sx}(t: s_2; s_3)=i_\cap^{\sx}(t: s_1; s_2)=-0.5849$ meaning that
    {\small
    \begin{align*}
        &\pi^{\sx}(t:\{2\}\{1,3\})=0\\
        &\pi^{\sx}(t:\{2\})=i^{\sx}(t: s_2) - i^{\sx}(t: s_2; s_1, s_3)= 0.5849.
    \end{align*}
    }
    Furthermore,
    {\small
    \begin{align*}
        \small
        &i_\cap^{\sx}(t: s_1; s_2, s_3)=i_\cap^{\sx}(t: s_1; s_1, s_2) = i_\cap^{\sx}(t: s_1) =0\\
        &i_\cap^{\sx}(t: s_3; s_1, s_2)= i_\cap^{\sx}(t: s_1; s_1, s_2) = i_\cap^{\sx}(t: s_1) =0\\
        &i_\cap^{\sx}(t: s_1,s_2; s_1, s_3; s_2,s_3)= i_\cap^{\sx}(t: s_1) =0
    \end{align*}
    }
    and so
    {\small
    \begin{align*}
        \small
        &\pi^{\sx}(t:\{1\}\{2,3\})=\pi^{\sx}(t:\{3\}\{1,2\})=0\\
        &\pi^{\sx}(t:\{1,2\}\{2,3\})=0.415\\
        &\pi^{\sx}(t:\{1,2\}\{1,3\})=\pi^{\sx}(t:\{1,2\}\{2,3\}) = 0.
    \end{align*}
    }
    Finally, we have
    {\small
    \begin{align*}
        \small
        &i_\cap^{\sx}(t: s_1, s_2, s_3)=i_\cap^{\sx}(t: s_1, s_2) = i_\cap^{\sx}(t: s_2, s_3) =1\\
        &i_\cap^{\sx}(t: s_1, s_3)=0
    \end{align*}
    }
    and thus it easy to see that their corresponding atoms are zero. 
    
    \subsection{Probability distribution 3-bit parity}
        Let $S_1,$ $S_2$ and $S_3$ be independent, uniformly distributed random bits, and $T=\sum_{i=1}^3 S_i \Mod 2$. This distribution is the 3-bit parity, where $T$ indicates the parity of the total number of 1-bits in $(S_1, S_2, S_3)$. Note that all possible realizations occur with probability $\nicefrac{1}{8}$ and result in the same PPID as well as the average PID due to the symmetry of the variables. Table~\ref{tab:3-bit-parity} shows the informative and misinformative component, and their difference for any realization. In addition, we illustrate in Figure~\ref{fig:4bit_parity_worked} the results of $\pi^{\sx}(t: \{1,2\}\{3,4\})$ for the 4-bit parity distribution.
        
        \begin{table*}
        \caption{\label{tab:3-bit-parity}\textbf{3-bit Parity Example.} \emph{Left:} the average \emph{informative} partial information decomposition  is evaluated. \emph{Right:} the average \emph{misinformative} partial information decomposition is evaluated. \emph{Center:} the average partial information decomposition is evaluated.}
        \centering\tiny
        \begin{tabular}{cccccccc}
            \toprule
            \multicolumn{4}{c}{$\Pi_+^{\sx}$}                &   \multicolumn{4}{c}{$\Pi_-^{\sx}$}\\
            \cmidrule(r){1-4}                                   \cmidrule(l){5-8}
            &$\{1,2,3\}$&&                                    &   &$\{1,2,3\}$&&\\         
            &0.2451   &&                                    &   &0        &&\\
            \cmidrule(r){1-4}                                   \cmidrule(l){5-8}
            $\{1,2\}$&$\{1,3\}$&$\{2,3\}$&                     &   $\{1,2\}$&$\{1,3\}$&$\{2,3\}$&\\         
            0.1699  &0.1699  &0.1699  &                     &   0       &0       &0       &\\
            \cmidrule(r){1-4}                                   \cmidrule(l){5-8}
            $\{1,2\}\{1,3\}$&$\{1,2\}\{2,3\}$&$\{1,3\}\{2,3\}$&   &   $\{1,2\}\{1,3\}$&$\{1,2\}\{2,3\}$&$\{1,3\}\{2,3\}$&\\         
            0.0931        &0.0931        &0.0931        &   &   0             &0             &0             &\\
            \cmidrule(r){1-4}                               \cmidrule(l){5-8}
            $\{1\}$&$\{2\}$&$\{3\}$&$\{1,2\}\{1,3\}\{2,3\}$    &   $\{1\}$&$\{2\}$&$\{3\}$&$\{1,2\}\{1,3\}\{2,3\}$\\         
            0.3219 &0.3219 &0.3219 &0.0182                  &   0.3219 &0.3219 &0.3219 &0.2451\\
            \cmidrule(r){1-4}                               \cmidrule(l){5-8}
            $\{1\}\{2,3\}$&$\{2\}\{1,3\}$&$\{3\}\{1,2\}$&      &   $\{1\}\{2,3\}$ &$\{2\}\{1,3\}$&$\{3\}\{1,2\}$&\\         
            0.0406        &0.0406       &0.0406       &     &   0.1699 &0.1699&0.1699&\\
            \cmidrule(r){1-4}                               \cmidrule(l){5-8}
            $\{1\}\{2\}$&$\{1\}\{3\}$&$\{2\}\{3\}$&         &   $\{1\}\{2\}$    &$\{1\}\{3\}$&$\{2\}\{3\}$&\\         
            0.2224       &0.2224      &0.2224      &        &   0.415           &0.415&0.415&\\
            \cmidrule(r){1-4}                               \cmidrule(l){5-8}
            &$\{1\}\{2\}\{3\}$&&                          &   &$\{1\}\{2\}\{3\}$  &&\\         
            &0.1926           &&                          &   &0                  &&\\
            \toprule
            &&\multicolumn{4}{c}{$\Pi^{\sx}$}\\
            \cmidrule(r){3-6}
            &&& $\{1,2,3\}$&&                                   &\\
            &&&0.2451&&                                       &\\
            \cmidrule(r){3-6}
            && $\{1,2\}$ &$\{1,3\}$&$\{2,3\}$&                   &\\
            &&0.1699 &0.1699&0.1699&                          &\\
            \cmidrule(r){3-6}
            &&$\{1,2\}\{1,3\}$ &$\{1,2\}\{2,3\}$&$\{1,3\}\{2,3\}$&  &\\
            &&0.0931 &0.0931&0.0931&                          &\\
            \cmidrule(r){3-6}
            &&$\{1\}$& $\{2\}$& $\{3\}$&$\{1,2\}\{1,3\}\{2,3\}$  &\\
            &&0.3219&0.3219&0.3219&-0.2268                    &\\
            \cmidrule(r){3-6}
            &&$\{1\}\{2,3\}$ &$\{2\}\{1,3\}$&$\{3\}\{1,2\}$&     &\\
            &&-0.1293 &-0.1293&-0.1293&                       &\\
            \cmidrule(r){3-6}
            &&$\{1\}\{2\}$ &$\{1\}\{3\}$&$\{2\}\{3\}$&        &\\
            &&-0.1926&-0.1926&-0.1926&                        &\\
            \cmidrule(r){3-6}
            &&&$\{1\}\{2\}\{3\}$&&                            &\\
            &&&0.1926                   &&                    &\\
            \bottomrule
        \end{tabular}
        \end{table*}

\begin{figure}
    \centering
    \includegraphics[width=0.35\textwidth]{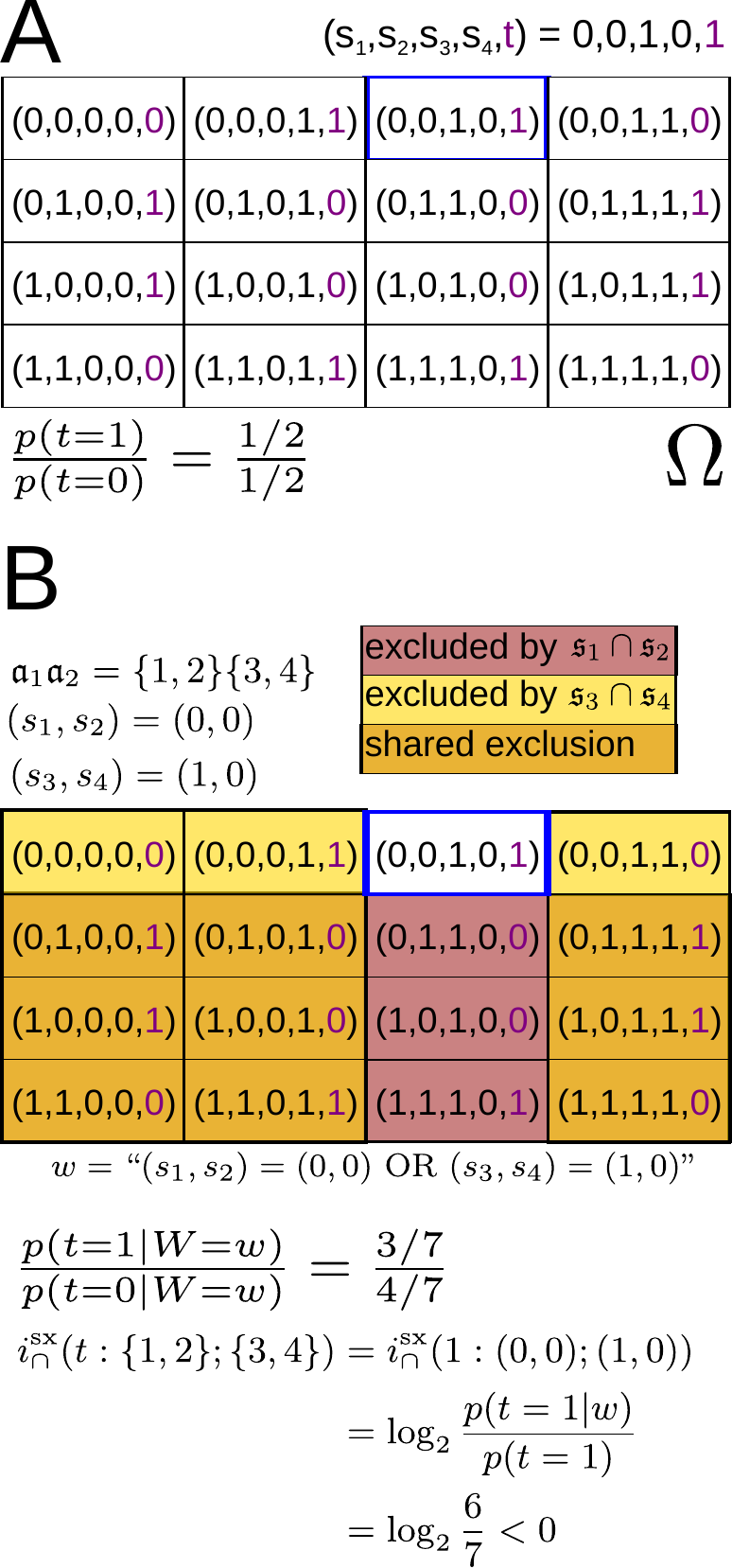}
    \caption{\label{fig:4bit_parity_worked}\textbf{Worked example of $i_\cap^{\sx}$ for a four source-variables case.}  We evaluate the shared information $i^{\sx}_\cap(t:\bfa_1;\bfa_2)$ with $\bfa_1=\{1,2\},$ $\bfa_2=\{3,4\}$, $s=(s_1,s_2,s_3,s_4) = (0,0,1,0),$ and $t=\mathrm{Parity}(s)=1$. (A) Sample space -- the relevant event is marked by the blue (gray) outline. (B) exclusions induced by the two collections of source realization indices $\bfa_1$ (brown (dark gray)), $\bfa_2$ (yellow (light gray)), and the shared exclusion relevant for $i^{\sx}_\cap$  (gold (gray)). After removing and rescaling, the probability for the target event that was actually realized, i.e., $t=1$, is reduced from $1/2$ to $3/7$. Hence the shared exclusion leads to negative shared information. Hence, $\pi^{\sx}(t:\{1,2\}\{3,4\})=-0.0145 \bit.$}
\end{figure}

\appendix
\section{\label{apx:sec:proofs}Lattice structure: supporting proofs and further details}
We show how the redundancy lattice can be endowed by $i_\cap^{\sx\pm}$ separately to obtain consistent PID terms $\pi^{\sx}_\pm$. Subsequently, we show that $\pi^{\sx}_\pm$ are nonnegative and thus the PID terms are meaningful.

\subsection{\label{apx:subsec:inf-misinf}Informative and misinformative lattices}
We start by explaining the redundancy lattice proposed by Williams and Beer. Then, we explain in detail how to apply $i_\cap^{\sx}$ to obtain a PID. 

As explained in section~\ref{sec:lattice}, there is a one-to-one correspondence between the PID terms and the antichain combinations. Since $i^{\sx}_\cap$ is defined locally, then for every realization the antichain combinations are associated to the source events. This way the PPID terms are computed and their average amount to the desired PID terms.

We use specific index sets and call them \emph{antichains} to represent the antichain combination since antichain combinations are uniquely identified by the indices of their source events. For instance, an antichain $\alpha = \{\bfa_1,\ldots,\bfa_n\}$ such that $\bfa_i \subset[n]$ where $[n]$ is the index set of the realization $s=(s_1,\ldots,s_n).$ Moreover, $\bfa_i\in\alpha$ should be pairwise incomparable under inclusion since antichain combinations are as such (see Section~\ref{sec:lattice}). E.g., $\{\{1,2\},\{1,3\}\}$ represents the source event $(\mfs_1\cap\mfs_2)\cup(\mfs_1\cap\mfs_3)$ and the combination of $(s_1,s_2)$ and $(s_1,s_3)$. 

Let $\mathscr A([n])$ be the set of all antichains; Crampton et al.~\cite{crampton2000embedding} showed that there exists the following partial ordering over $\mathscr A([n])$: 
\[
    \alpha\preceq\beta\Leftrightarrow\forall~\bfb\in\beta, \exists~\bfa\in\alpha\mid \bfa\subseteq\bfb \quad \forall~\alpha,\beta\in\mathscr A ([n]).
\]
This partial ordering $\preceq$ implies that any $\alpha,\beta\in\mathscr A([n])$ have an infimum $\alpha\wedge\beta\in\mathscr A([n])$ and a supremum $\alpha\vee\beta\in\mathscr A([n])$ and so $\left<\mathscr{A}([n]), \preceq\right>$ is called a lattice. Now when endowing $\left<\mathscr{A}([n]), \preceq\right>$ with a function $f$ (say a shared information) such that $f(\alpha)=\sum_{\beta\preceq\alpha}\pi(\beta)$ where $\pi(\beta)$ are desired quantities (say PID terms) that have a one-to-one correspondence with $\beta\in\mathscr A([n])$, then we can compute these $\pi$ using $f$. Hence, we reduced the problem of defining different conceptual quantities that each antichain represents by defining a single conceptual quantity for each antichain that is the shared mutual information.

Williams and Beer coined this idea of endowing $\left<\mathscr{A}([n]), \preceq\right>$ with a redundancy measure $I_\cap$ and hence the name ``redundancy lattice.'' For this, they had a set of axioms that ensured (i) the one-to-one correspondence between $\mathscr A([n])$ and the PID terms and (ii) that $I_\cap(\alpha)=\sum_{\beta\preceq\alpha}\Pi(\beta).$ However, their definition was not local (for every realization) and thus Finn and Lizier~\cite{finn2018pointwise} adapted the axioms for the local case. However, the local shared measure $i_\cap$ can take negative values and the problem persists upon averaging. Thus, they proposed to decompose $i_\cap = i_\cap^+ - i_\cap^-$ where $i_\cap^\pm$ take only nonnegative terms and can be interpreted as informative and misinformative components of $i_\cap.$ Altogether, for each realization we will endow $\left<\mathscr{A}([n]), \preceq\right>$ with $i_\cap^{\sx+}$ (\emph{informative lattice}) and $i_\cap^{\sx-}$ (\emph{misinformative lattice}) individually to obtain $\pi^{\sx}_+$ and $\pi^{sx}_-$ PPID terms. 

First, for any $\alpha\in\mathscr A[n]$, we define $i_\cap^{\sx\pm}$ as follows: 
\[
\begin{split}
    \mP(\alpha) &= \mP(\bigcup_{\bfa\in\alpha}\bigcap_{i\in\bfa}\mfs_i)\\
    \mP(\mft,\alpha) &= \mP(\bigcup_{\bfa\in\alpha}\bigcap_{i\in\bfa}(\mft\cap\mfs_i))\\
    i^ {\sx}_{\cap}(t:\alpha)    &= \log_2\frac{1}{\mP(\alpha)} - \log_2\frac{\mP(\mft)}{\mP(\mft\cap\alpha)}\\
                                &= i^ {\sx+}_{\cap}(t:\alpha) - i^ {\sx-}_{\cap}(t:\alpha).
\end{split}
\]
Now to show that this endowing of $i_\cap^{\sx\pm}$ is consistent, we prove Theorem~\ref{thm:axioms}, that shows that $i_\cap^{\sx\pm}$ satisfy the PPID axioms. 

\begin{proof}[proof of Theorem~\ref{thm:axioms}]
    By the symmetry of intersection, $i^ {\sx\pm}_{\cap}$ defined in~\eqref{eq:isx_decomposed} satisfy the symmetry Axiom~\ref{ax:sym}. For any collection $\SE a,$ using~\eqref{eq:isx_decomposed}, the informative and misinformative shared information are
    \begin{align*}
        i^{\sx+}_\cap(t:\SE a) &= \log_2\frac{1}{p(\bfa)} = h(\bfa)\\
        i^{\sx-}_\cap(t:\SE a) &= \log_2\frac{p(t)}{p(t,\bfa)} = h(\bfa \mid t).
    \end{align*} 
    and so they satisfy Axiom~\ref{ax:self}. For Axiom~\ref{ax:mono}, note that
    \[
        \mP(\bmfa_1, \bmfa_2, \ldots, \bmfa_m, \bmfa_{m+1})\leq \mP(\bmfa_1,\bmfa_2, \ldots, \bmfa_m)
    \]
    This implies that $i^ {\sx\pm}_{\cap}$ decrease monotonically if joint source realizations are added, where equality holds if there exists $i\in[m]$ such that $\bmfa_{m+1}\supseteq\bmfa_i$ , i.e., if there exists $i\in[m]$ such that $\mfa_{m+1}\subseteq\mfa_i \Leftrightarrow \bfa_i\subseteq\bfa_{m+1}.$ 
\end{proof}
Then, we assume that
\begin{equation}\label{eq:sx-lat}
   i^{\sx\pm}_{\cap}(t:\alpha) = \sum_{\beta\preceq\alpha}\pi_\pm^{\sx}(t:\beta)\quad\forall~\alpha,\beta\in\mathscr A([n]).
\end{equation}
Note that, this assumption is logically sound and is discussed thoroughly in~\cite{gutknecht2020bits}. Finally, to obtain $\pi^{\sx}_\pm,$ we show that Eq.~\eqref{eq:sx-lat} is invertible via a so-called M\"{o}bius inversion given by the following theorem. \begin{theorem}\label{thm:moeb-inv}
    Let $i^ {\sx\pm}_{\cap}$ be measures on the redundancy lattice, then we have the following closed form for each atom $\pi_\pm^{\sx}$:
    \begin{equation}\label{eq:moeb-inv}
        \pi_\pm^{\sx}(t:\alpha) = i^{\sx\pm}_{\cap}(t:\alpha) - \sum_{\emptyset\neq\mathcal B\subseteq\alpha^-}(-1)^{|\mathcal B|-1}i^{\sx\pm}_{\cap}(t:\bigwedge\mathcal B).
    \end{equation}
\end{theorem}
\noindent The proof of the above theorem follows from that of~\cite[Theorem A1]{finn2018pointwise}. 

\subsection[]{\label{apx:subsec:nonneg}Nonnegativity of $\pi_\pm^{\sx}$}
    In order for our information decomposition to be interpretative, the informative and misinformative atoms, $\pi_\pm^{\sx}$, must be nonnegative. First, we recall these results from convex analysis that will come in handy later. 
    \begin{theorem}[Theorem 2.67~\cite{ruszczynski2006nonlinear}]\label{thm:1st-ord-convex}
        Let $f:\RR^n\to\RR$ be a continuously differentiable function. Then, $f$ is convex if and only if for all $x$ and $y$
        \[
            f(y)\ge f(x) + \nabla^T f(x)(y-x).
        \]  
    \end{theorem}
    \begin{proposition}\label{prop:conv-dec-gen}
        Let $f:\RR^n\to\RR$ be a continuously differentiable convex function and $y_0 - x_0= c\mathbf{1}$ where $c\ge 0$. If $f(x_0)\ge f(y_0)$, then
        \[
            -\sum_i \frac{\partial f}{\partial x_i}(y_0) \le -\sum_i \frac{\partial f}{\partial x_i}(x_0). 
        \]
    \end{proposition}
    \begin{proof}
        For any $x,y\in\RR^n$, using theorem~\ref{thm:1st-ord-convex} by interchanging the roles of $x$ and $y$,
       \begin{equation*}
            -\nabla^T f(y) (y-x) \le f(x) - f(y) \le -\nabla^T f(x)(y-x).
        \end{equation*}
        Now consider $x_0,y_0\in\RR^n$ such that $y_0 - x_0 = c\mathbf{1},$ then
        \[
            \begin{split}
            -c\nabla^Tf(y_0)\mathbf{1}&\le -c\nabla^Tf(x_0)\mathbf{1}\\
                -\sum_i \frac{\partial f}{\partial x_i}(y_0) &\le -\sum_i \frac{\partial f}{\partial x_i}(x_0).
            \end{split}
        \]
    \end{proof}
    We write down the proof of theorem~\ref{thm:mono-antichain} and then show that $i_\cap^{\sx\pm}$ are nonnegative.  
    
    \begin{proof}[proof of theorem~\ref{thm:mono-antichain}]
    Let $\alpha, \beta \in \mathscr{A}([n])$ and $\alpha \preceq \beta$. Then $\alpha$ and $\beta$ are of the form $\alpha = \{\mathbf{a}_1,....,\mathbf{a}_{k_\alpha}\}$ and $\beta = \{\mathbf{b}_1,....,\mathbf{b}_{k_\beta}\}$. Because $\alpha \preceq \beta$ there is a function $f: \beta \rightarrow \alpha$ such that $f(\mathbf{b}) \subseteq \mathbf{b}$ \footnote{This function does not have to be surjective: Suppose $\alpha = \{\{1\}, \{2,4\}, \{3\}\}$ and $\beta = \{\{1,2,3,4\}\}$. Then necessarily two sets in $\alpha$ will not be in the image of $f.$ It also does not have to be injective. Consider $\alpha = \{1\}$ and $\beta = \{ \{1,2\}, \{1,3\}\}$. Then both elements of $\beta$ have to be mapped to the only element of $\alpha$}. Now we have for all $\mathbf{b} \in \beta$
    \begin{equation*}
        \bigcap_{i \in \mathbf{b}} \mfs_i \subseteq \bigcap_{i \in f(\mathbf{b})}\mfs_i
    \end{equation*}
    Hence,
    \begin{equation}\label{eq:prob-alpha-beta}
        \small
        \begin{split}
            \mathbb{P}(\beta)   &= \mathbb{P}\left(\bigcup_{\mathbf{b} \in \beta} \bigcap_{i \in \mathbf{b}}\mfs_i\right) \leq \mathbb{P}\left(\bigcup_{\mathbf{b} \in \beta} \bigcap_{i \in f(\mathbf{b})}\mfs_i\right)\\
                                &\leq  \mathbb{P}\left(\bigcup_{\mathbf{a} \in \alpha} \bigcap_{i \in \mathbf{a}}\mfs_i\right) = \mathbb{P}(\alpha).
        \end{split}
    \end{equation}
    The last inequality is true because the term on its L.H.S.~is the probability of a union of intersections related to  collections $\mathbf{a} \in \alpha$ (the $f(\mathbf{b})$), i.e., it is the probability of a union of events of the type $\bigcap_{i \in \mathbf{a}} \mfs_i$. The probability of such a union can only get bigger if we take it over \textit{all} events of this type. Using~\eqref{eq:prob-alpha-beta}, it immediately follows that $i^ {\sx+}_{\cap}(t:\alpha) \le i^{\sx+}_{\cap}(t:\beta)$ and $i^ {\sx+}_{\cap}$ is monotonically increasing. Using the same argument, $i^ {\sx-}_{\cap}$ is monotonically increasing.
    \end{proof}
    \begin{proposition}
        $i^ {\sx\pm}_{\cap}$ are nonnegative.
    \end{proposition}
    \begin{proof}
        $i^ {\sx+}_{\cap}(t:\bfa_1; \bfa_2; \ldots; \bfa_m)= \log_2\frac{1}{\mP(\mfa_1\cup\mfa_2\cup\ldots\cup\bfa_m)}\ge 0.$ 
        
        \noindent Similarly, the misinformative
        $i^{\sx-}_{\cap}(t:\bfa_1; \bfa_2;\ldots; \bfa_m) = \log_2\frac{\mP(\mft)}{\mP(\mft\cap[(\cap_{i\in\bfa_1}\mfs_i)\cup(\cap_{i\in\bfa_2}\mfs_i)\cup\ldots\cup(\cap_{i\in\bfa_m}\mfs_i)])}\ge 0.$
    \end{proof}
    We construct a family of mappings from $\mathscr P(\alpha^-)$ where $\alpha^-$ is the set of children of $\alpha$ to the $\mathscr A([n])$ (see FIG~\ref{fig:Lattice-map}). This family of mappings plays a key role in the desired proof of nonnegativity.
    \begin{proposition}\label{prop:the-mapping}
        Let $\alpha\in\mathscr A([n])$ and $\alpha^- = \{\gamma_1, \dots, \gamma_k\}$ ordered increasingly w.r.t.~the probability mass be the set of children of $\alpha$ on $\left<\mathscr A([n]), \preceq\right>.$ Then, for any $1\le i\le k$
        \begin{alignat*}{2}
            f_i: &\mathscr P_1(\alpha^-\backslash\{\gamma_i\})\cup\{\{\alpha\}\} &\longrightarrow{}&\mathscr A([n])\\
            &\qquad\mathcal B                               &\longrightarrow{}&\qquad\bigwedge_{\beta\in\mathcal B}\beta\wedge\gamma_i
        \end{alignat*}
        is a mapping such that $\mP(f_i(\mathcal B)) = \mP(\bigwedge_{\beta\in\mathcal B}\beta) + d_i$ where $d_i = \mP(\gamma_i) - \mP(\alpha)$ and the complement is taken w.r.t.~$\mathscr P(\alpha^-),$ the powerset of $\alpha^-.$
    \end{proposition}
    \begin{proof}
          Since $\gamma_i\in\alpha^-$ and  $\beta\in\alpha^- $ for any $\beta\in\mathcal B,$ then $(\bigwedge_{\beta\in\mathcal B}\beta)\vee\gamma_i = \alpha$. Now, for any $\mathcal B\in\mathscr P(\alpha^-\backslash\{\gamma_i\}),$ using the inclusion-exclusion, $\beta\wedge\gamma_i = \underline{\beta\cup\gamma_i}$ and $\beta\vee\gamma_i = \underline{\uparrow\beta\cap\uparrow\gamma_i},$
        \[\small
            \begin{split}
                \mP(f_i(\mathcal B)) = \mP(\bigwedge_{\beta\in\mathcal B}\beta\wedge\gamma_i) &= \mP(\bigwedge_{\beta\in\mathcal B}\beta) + \mP(\gamma_i) - \mP(\bigwedge_{\beta\in\mathcal B}\beta\vee\gamma_i)\\ 
                &= \mP(\beta) + \mP(\gamma_i) - \mP(\alpha).
            \end{split}
        \]
    \end{proof}
    \begin{figure}
        \centering
        \includegraphics[width=0.3\textwidth]{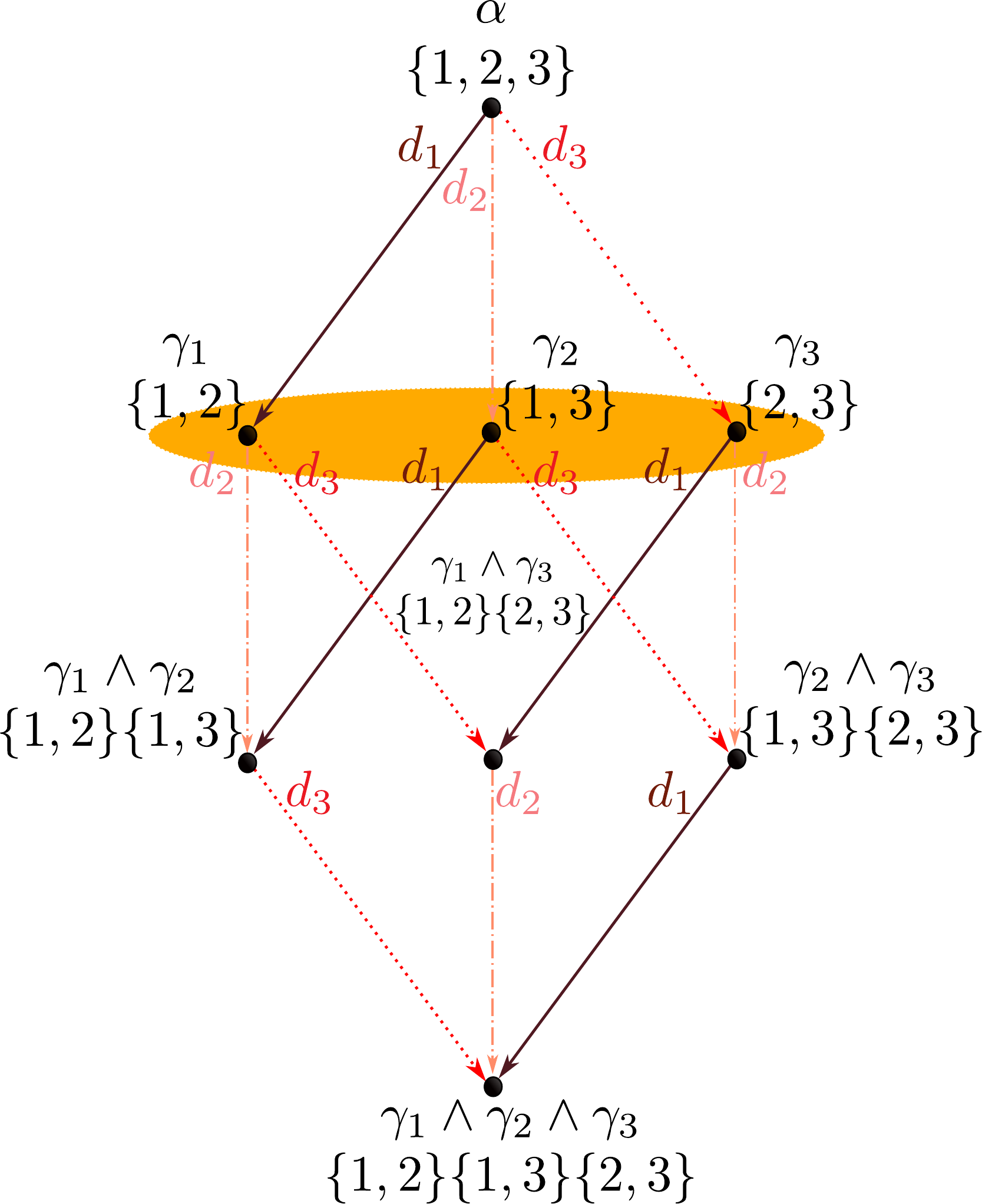}
        \caption{\textbf{The family of mappings introduced in proposition~\ref{prop:the-mapping} that preserve the probability mass difference.} Let $\alpha$ be the top node of $\mathscr A([3]).$ The orange (gray dotted) region is $\alpha^-,$ the set of children of $\alpha$. Each color depicts one mapping in the family based on some $\gamma\in\alpha^-$. The dark red (solid line) mapping is based on $\gamma_1,$ the red mapping (dash-dotted line) is based on $\gamma_2$ and the salmon (dotted line) mapping is based on $\gamma_3.$}
        \label{fig:Lattice-map}
    \end{figure}
    
   \begin{figure}
        \centering
        \includegraphics[width=0.3\textwidth]{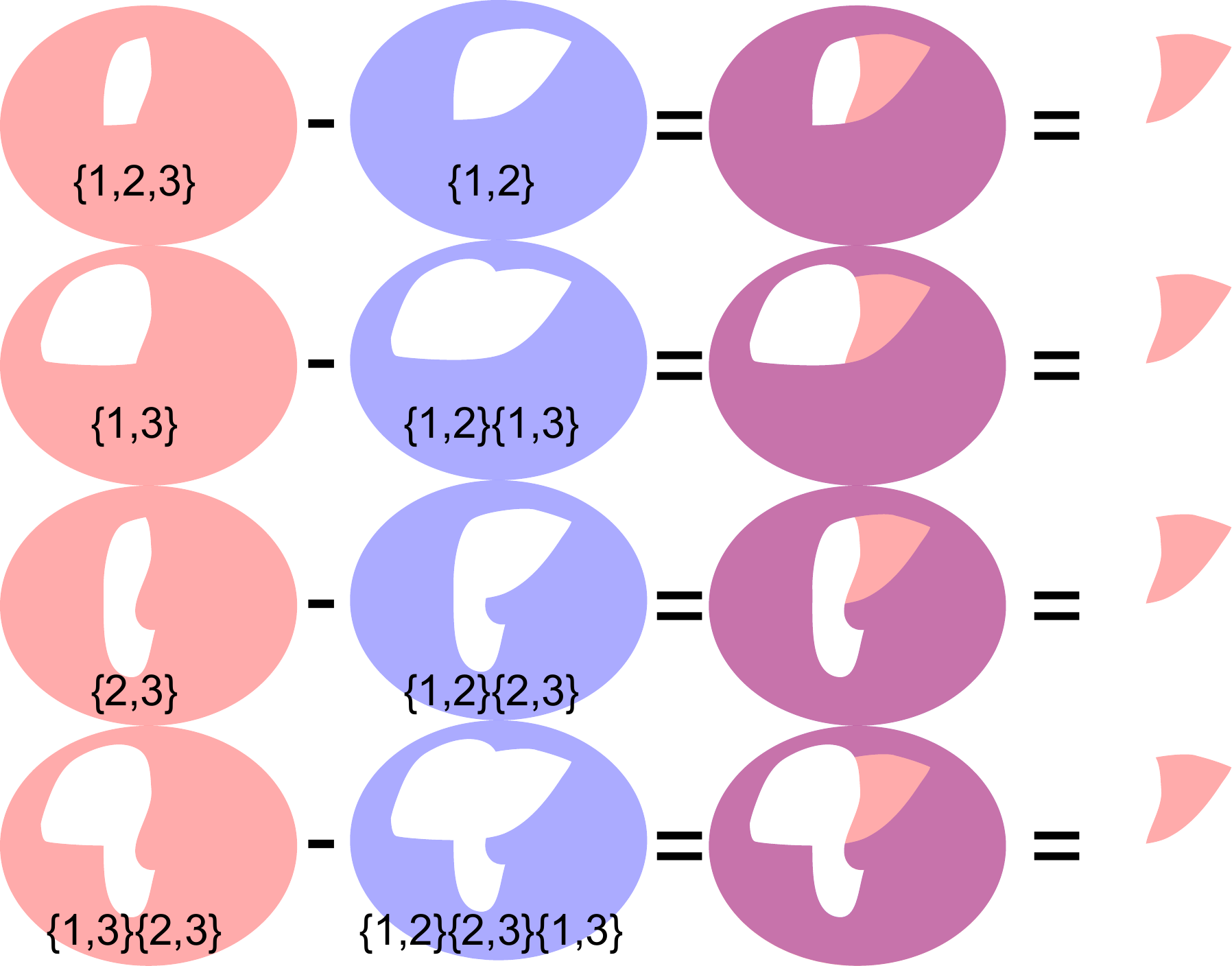}
        \caption{\textbf{Depiction of set differences corresponding to the probability mass difference $d_1$ introduced in proposition~\ref{prop:the-mapping} and shown in Fig.~\ref{fig:Lattice-map}, for the sets from Fig.~\ref{fig:Exclusion_lattice}.}
        \label{fig:d1}}
    \end{figure}
    The following lemma shows that for any node $\alpha\in\mathscr A([n])$, the recursive Eq.~\eqref{eq:moeb-inv} should be nonnegative which is the main point in the desired proof of nonnegativity.
    \begin{lemma}\label{lem:log-ineq}
        Let $\alpha\in \mathscr A([n]);$ then
        \begin{equation}\label{eq:log-ineq}
            -\log_2 \mP(\alpha) + \sum_{\emptyset\neq\mathcal B\subseteq\alpha^-}(-1)^{|\mathcal B| - 1}\log_2\mP(\bigwedge\mathcal B) \ge 0.
        \end{equation}
    \end{lemma}
    \begin{proof}
        Suppose that $|\alpha^-| = k$ and w.l.o.g.~that $\alpha^-=\{\gamma_1,\dots,\gamma_k\}$ is ordered increasingly w.r.t.~the probability mass. The proof will follow by induction over $k=|\alpha^-|$. We will demonstrate the inequality~\eqref{eq:log-ineq} for $k=3,4$ to show the induction basis. For $k=3,$ the L.H.S.~of~\eqref{eq:log-ineq} can be written as
        \[\small
        \begin{split}
            &\log_2\frac{\mP(\gamma_1)\mP(\gamma_2)\mP(\gamma_3)\mP(\gamma_1\wedge\gamma_2\wedge\gamma_3)}{\mP(\alpha)\mP(\gamma_1\wedge\gamma_2)\mP(\gamma_1\wedge\gamma_3)\mP(\gamma_2\wedge\gamma_3)}\\
            &= \log_2\frac{\frac{\mP(\alpha) + d_1}{\mP(\alpha)}}{\frac{(\mP(\alpha) + d_2) + d_1}{(\mP(\alpha) + d_2)}} - \log_2\frac{\frac{\mP(\alpha) + d_3 + d_1}{\mP(\alpha) + d_3}}{\frac{(\mP(\alpha) + d_3 + d_2) + d_1}{(\mP(\alpha) + d_3 + d_2)}}\\
            &= [h_3(\mP(\alpha)) - h_3(\mP(\alpha) + d_2)]\\
            &- [h_3(\mP(\alpha) + d_3) - h_3(\mP(\alpha) + d_3 + d_2)],
        \end{split}
        \]
        where $h_3(x)=\log_2(1 + \nicefrac{d_1}{x}),$ $d_i:=\mP(\gamma_i) - \mP(\alpha)$ for $i\in\{1,2,3\},$ and $d_3\ge d_2\ge d_1\ge 0.$ Note that $h_3$ is a continuously differentiable convex function that is monotonically decreasing. Now, take $x = \mP(\alpha)$ and $y=\mP(\alpha) + d_3,$ then
        \[
            \begin{split}
                &h_3(\mP(\alpha)) - h_3(\mP(\alpha) + d_2)\\
                &\stackrel{\text{Thm.}~\ref{thm:1st-ord-convex}}{\ge}-d_2h_3'(\mP(\alpha) + d_2)\\
                &\stackrel{\text{Prop.}~\ref{prop:conv-dec-gen}}{\ge} -d_2h_3'(\mP(\alpha) + d_3)\\
                &\stackrel{\text{Thm.}~\ref{thm:1st-ord-convex}}{\ge} h_3(\mP(\alpha) + d_3)-h_3(\mP(\alpha) + d_3 + d_2)
            \end{split}
        \]
        and so the inequality~\eqref{eq:log-ineq} holds when $k=3.$ For $k=4,$ we have $\alpha^-=\{\gamma_1,\gamma_2,\gamma_3,\gamma_4\}$ ordered increasingly w.r.t.~the probability mass. By Proposition~\ref{prop:the-mapping}, the L.H.S.~of~\eqref{eq:log-ineq} can be written as
        \[\small
        \begin{split}
            &\bigg[ h_3\big(\mP(\alpha)\big) - h_3\big(\mP(\alpha) + d_2\big) - \bigg( h_3(\mP(\alpha) + d_3)\\ 
            &- h_3\big(\mP(\alpha) + d_3 + d_2\big) \bigg) \bigg]\\ 
            &-\bigg[ h_3\big(\mP(\alpha) + d_4\big) - h_3\big(\mP(\alpha) + d_4 + d_2\big)\\ 
            &- \bigg( h_3\big(\mP(\alpha) + d_4 + d_3\big) - h_3\big(\mP(\alpha) +d_4 + d_3 + d_2\big) \bigg) \bigg]\\
            &= \bigg[ h_4\big(\mP(\alpha), \mP(\alpha) + d_2\big) - h_4\big(\mP(\alpha) + d_3, \mP(\alpha) + d_3 + d_2\big) \bigg]\\
            &- \bigg[ h_4\big(\mP(\alpha) + d_4, \mP(\alpha) + d_4 + d_2\big)\\ 
            &- h_4\big(\mP(\alpha) + d_4 + d_3, \mP(\alpha) + d_4 + d_3 + d_2\big) \bigg],
        \end{split}
        \]
        where $d_i:=\mP(\gamma_i) - \mP(\alpha)$ for $i\in\{2,3,4\},$ $d_4\ge d_3\ge d_2\ge 0,$ and $h_4(x_1,x_2) = \log_2(1 +\nicefrac{d_1(x_2 - x_1)}{x_1(x_2 + d_1)}) = h_3(x_1) - h_3(x_2).$ Let $\delta\ge 0$ and  $x,y\in H_4^\delta:=\{x\in\RR^{2*}_+\mid x_2 = x_1 + \delta\}$ where $x_1\le y_1,$ then $h_4(x)\ge h_4(y)$ since~\eqref{eq:log-ineq} holds for $k=3.$ Moreover, $h_4$ is convex since for any $x,y\in H_4^\delta$ and $\theta\in[0,1]$
        \[\small
            \begin{split}
                &\theta h_4(x) + (1-\theta)h_4(y) - h_4(\theta x + (1-\theta)y)\\
                &= \theta (h_3(x_1) - h_3(x_2)) + (1-\theta)(h_3(y_1)\\ 
                &- h_3(y_2)) - h_3(\theta x_1 + (1-\theta)y_1) + h_3(\theta x_2 + (1-\theta)y_2)\\
                &= [\theta h_3(x_1) + (1-\theta)h_3(y_1) -  h_3(\theta x_1 + (1-\theta)y_1)]\\
                &- [\theta h_3(x_1 + \delta) + (1-\theta)h_3(y_1 + \delta)\\ 
                &- h_3(\theta x_1 + (1-\theta)y_1 + \delta)]\ge 0.
            \end{split}
        \]
         Now, take $x =(\mP(\alpha),\mP(\alpha) + d_2)$ and $y=(\mP(\alpha) + d_4, \mP(\alpha) + d_4 + d_2),$ then 
        \[\small
            \begin{split}
                &h_4(\mP(\alpha),\mP(\alpha) + d_2) - h_3(p(\alpha) + d_3, \mP(\alpha) + d_3 + d_2)\\
                &\stackrel{\text{Thm.}~\ref{thm:1st-ord-convex}}{\ge} -\nabla^Th_4(\mP(\alpha) + d_3, \mP(\alpha) + d_3 + d_2)(d_3,d_3)\\
                &\stackrel{\text{Prop.}~\ref{prop:conv-dec-gen}}{\ge} -\nabla^Th_4(\mP(\alpha) + d_4, \mP(\alpha) + d_4 + d_2)(d_3,d_3)\\
                &\stackrel{\text{Thm.}~\ref{thm:1st-ord-convex}}{\ge} h_4(\mP(\alpha) + d_4, \mP(\alpha) + d_4 + d_2)\\ &
                - h_4(\mP(\alpha) + d_4 + d_3, \mP(\alpha) + d_4 + d_3 + d_2),       
            \end{split}
        \]
        and so the inequality~\eqref{eq:log-ineq} holds.  
        
        Suppose that the inequality holds for $k$ and let us proof it for $k+1$. Here $\alpha^- =\{\gamma_1,\gamma_2,\ldots,\gamma_{k+1}\}$ and using Proposition~\ref{prop:the-mapping}, the L.H.S.~of~\eqref{eq:log-ineq} can be written as
        \[\small
            \begin{split}
                &\bigg[h_k\big(a_{k-2}\big) - h_k\big(a_{k-2} + d_{k-1}\mathbf{1}_{k-2}\big) -\bigg(h_k\big(a_{k-2} + d_{k}\mathbf{1}_{k-2}\big)\\
                &- h_k\big(a_{k-2} + (d_k +d_{k-1})\mathbf{1}_{k-2}\big)\bigg)\bigg]\\
                &-\bigg[h_k\big(a_{k-2} + d_{k+1}\mathbf{1}_{k-2}\big) - h_k\big(a_{k-2} + (d_{k+1} + d_{k-1})\mathbf{1}_{k-2}\big)\\ 
                &-\bigg(h_k\big(a_{k-2} + (d_{k+1} + d_{k})\mathbf{1}_{k-2}\big)\\ 
                &- h_k\big(a_{k-2} + (d_{k+1} + d_k + d_{k-1})\mathbf{1}_{k-2}\big)\bigg)\bigg]\\
                &=\bigg[h_{k+1}\big(a_{k-2}, a_{k-2} + d_{k-1}\mathbf{1}_{k-2}\big)\\ 
                &-\bigg(h_{k+1}\big(a_{k-2} + d_{k}\mathbf{1}_{k-2}, a_{k-2} + (d_k +d_{k-1})\mathbf{1}_{k-2}\big)\bigg)\bigg]\\
                &-\bigg[h_{k+1}\big(a_{k-2} + d_{k+1}\mathbf{1}_{k-2}, a_{k-2} + (d_{k+1} + d_{k-1}\big)\mathbf{1}_{k-2})\\ 
                &-h_{k+1}\big(a_{k-2} + (d_{k+1} + d_{k})\mathbf{1}_{k-2}, a_{k-2} + (d_{k+1} + d_k\\
                &+ d_{k-1})\mathbf{1}_{k-2}\big)\bigg]
            \end{split}
        \]
        where $a_{k-2} := (\mP(\alpha),\dots,\mP(\alpha) + \sum_{i=2}^{k-2}d_i)\in\RR^{2^{k-2}},$ $d_i:=\mP(\gamma_i) - \mP(\alpha)$ for $i\in\{2,\dots,k+1\},$ $d_{k+1}\ge\dots\ge d_2\ge 0,$ and $h_{k+1}(x_1,\dots, x_{2^{k-1}}) = h_{k}(x_1,\dots, x_{2^{k-2}}) - h_k(x_{2^{k-2}+1},\dots, x_{2^{k-1}}).$ 
        
        Let $\delta \ge 0$ and $x,y\in H_{k+1}^\delta:=\{x\in\RR^{2^{k-1}}\mid x_i = x_{j} + \delta,i = j \Mod~2^{k-2}\}$ where $x_i\le y_i$ for all $i,$ then $h_{k+1}(x)\ge h(y)$ because the Ineq.~\eqref{eq:log-ineq} holds for $k.$ Moreover, $h_{k+1}$ is convex since for any $x,y\in H_{k+1}^\delta$ and $\theta\in[0,1]$
        \[\small
            \begin{split}
                &\theta h_{k+1}(x_1,\dots, x_{2^{k-1}}) + (1-\theta) h_{k+1}(y_1,\dots, y_{2^{k-1}})\\ 
                &- h_{k+1}(\theta x_1 +(1-\theta)y_1 ,\dots, \theta x_{2^{k-1}} + (1-\theta)y_{2^{k-1}})\\
                &= \bigg[\theta h_k(x_1,\dots, x_{2^{k-2}}) + (1-\theta)h_k(y_1,\dots, y_{2^{k-2}})\\ 
                &- h_k(\theta x_1 +(1-\theta)y_1 ,\dots, \theta x_{2^{k-1}} + (1-\theta)y_{2^{k-2}})\bigg] - \bigg[\\
                &\theta h_k(x_1 +\delta,\dots, x_{2^{k-2}} +\delta) + (1-\theta)h_k(y_1 + \delta,\dots, y_{2^{k-2}} + \delta)\\ 
                &- h_k(\theta x_1 + (1-\theta)y_1 + \delta ,\dots, \theta x_{2^{k-2}} + (1-\theta)y_{2^{k-2}} + \delta)\bigg].
            \end{split}
        \]
        is nonnegative. Now, take $x =(a_{k-2},a_{k-2} + d_{k-1}\mathbf{1}_{k-2})$ and $y=(a_{k-2} + d_{k+1}\mathbf{1}_{k-2}, a_{k-2} + (d_{k+1} + d_{k-1})\mathbf{1}_{k-2}),$ then 
        \[\small
            \begin{split}
                &h_{k+1}(a_{k-2}, a_{k-2} + d_{k-1}\mathbf{1}_{k-2}) -h_{k+1}(a_{k-2} + d_{k}\mathbf{1}_{k-2},\\ 
                &a_{k-2} + (d_k +d_{k-1})\mathbf{1}_{k-2})\\
                &\ge -d_{k}\nabla^Th_{k+1}(a_{k-2} + d_{k}\mathbf{1}_{k-2}, a_{k-2} + (d_k +d_{k-1})\mathbf{1}_{k-2})\mathbf{1}_{k-1}\\
                &\ge -d_{k}\nabla^Th_{k+1}(a_{k-2} + d_{k+1}\mathbf{1}_{k-2},\\ 
                &a_{k-2} + (d_{k+1} +d_{k-1})\mathbf{1}_{k-2})\mathbf{1}_{k-1}\\
                &\ge h_{k+1}(a_{k-2} + d_{k+1}\mathbf{1}_{k-2}, a_{k-2} + (d_{k+1} + d_{k-1})\mathbf{1}_{k-2})\\
                &-h_{k+1}(a_{k-2} + (d_{k+1} + d_{k})\mathbf{1}_{k-2},\\ 
                &a_{k-2} + (d_{k+1} + d_k + d_{k-1})\mathbf{1}_{k-2}),
            \end{split}
        \]
        where the first and third inequalities hold using theorem~\ref{thm:1st-ord-convex} and the second inequality holds using Proposition~\ref{prop:conv-dec-gen} and so the inequality~\eqref{eq:log-ineq} holds for $k+1.$
    \end{proof}
    Finally we write down the proof of theorem~\ref{thm:non-neg} to conclude that  $i_\cap^{\sx}$ yields meaningful PPID terms.
    \begin{proof}[proof of theorem~\ref{thm:non-neg}]
        For any $\alpha\in\mathscr A([n]),$
        \[\small
            \begin{split}
                \pi_{+}^{\sx}(t:\alpha)   &= i^{\sx+}_{\cap}(t:\alpha) - \sum_{\emptyset\neq\mathcal B\subseteq\alpha^-}(-1)^{|\mathcal B|-1}i^ {\sx+}_{\cap}(t:\bigwedge\mathcal B)\\
                                    &= -\log_2 \mP(\alpha) + \sum_{\emptyset\neq\mathcal B\subseteq\alpha^-}(-1)^{|\mathcal B| - 1}\log_2 \mP(\bigwedge\mathcal B).
            \end{split}
        \]
        So, by Lemma~\ref{lem:log-ineq} $\pi^{\sx}_+(t:\alpha)\ge 0.$ Similarly, $\pi^{\sx}_-(t:\alpha)\ge 0$ since intersecting with $t$ has no effect on the nonnegativity shown in Lemma~\ref{lem:log-ineq}.
    \end{proof}

\section[]{\label{apx:sec:axiomatic_probability_theory}Definition of $i_\cap^{\sx}$ starting from a general probability space}
Let $(\Omega, \mathfrak{A}, \mathbb{P})$ be a probability space and $S_1, ..., S_n, T$ be discrete and finite random variables on that space, i.e., 
\begin{align*}
    \small
    &S_i: \Omega \rightarrow \mathcal A_{S_i}, \hspace{0.3cm }(\mathfrak{A},\mathscr{P}(\mathcal A_{S_i}))-\text{measurable} \\
    &T: \Omega \rightarrow \mathcal A_{T}, \hspace{0.3cm }(\mathfrak{A},\mathscr{P}(\mathcal A_{T}))-\text{measurable},
\end{align*}
where $\mathcal{A}_{S_i}$ and $\mathcal{A}_T$ are the finite alphabets of the corresponding random variables and $\mathscr{P}(\mathcal A_{S_i})$ and $\mathscr{P}(\mathcal A_{T}) $ are the power sets of these alphabets. Given a subset of source realization indices $\mathbf{a} \subseteq \{1,...,n\}$ the \emph{local mutual information} of source realizations $(s_i)_{i \in \mathbf{a}}$ about the target realization $t$ is defined as
\begin{equation*}
    \small
    i(t:(s_i)_{i \in \mathbf{a}})  = i(t:\mathbf{a}) = \log_2\frac{\mathbb{P}\left(\mathfrak{t}|\bigcap_{i \in \mathbf{a}} \mathfrak{s}_i \right)}{\mathbb{P}(\mathfrak{t})}.
\end{equation*}

The \emph{local shared information} of an antichain $\alpha = \{\mathbf{a}_1 , \dots, \mathbf{a}_m\}$ (representing a set of collections of source realizations) about the target realization $t \in \mathcal{A}_T$ is  defined in terms of the original probability measure $\mathbb{P}$ as a function $i_\cap^{\sx}: \mathcal{A}_T \times \mathscr{A}(s) \rightarrow \mathbb{R}$ with
\begin{equation*}
    \small
    i_\cap^{\sx}(t:\alpha) = i_\cap^{\sx} (t:\mathbf{a}_1; \dots; \mathbf{a}_m) := \log_2\frac{\mathbb{P}\left(\mathfrak{t}|\bigcup_{i=1}^m \mathfrak{a}_i \right)}{\mathbb{P}(\mathfrak{t})}.
\end{equation*}
A special case of this quantity is the local shared information of a \emph{complete} sequence of source realizations $(s_1, \dots, s_n)$ about the target realization $t$. This is obtained by setting $\mathbf{a}_i = \{i\}$ and $m=n$:
\begin{equation*}
    \small
    i_\cap^{\sx} (t:\{1\};\dots;\{n\}) = \log_2\frac{\mathbb{P}\left(\mathfrak{t}|\bigcup_{i=1}^n \mathfrak{s}_i\right)}{\mathbb{P}(\mathfrak{t})}.
\end{equation*}
In contrast to other shared information terms, this is an \emph{atomic} quantity corresponding to the very bottom of the lattice of antichains. Rewriting $i_\cap^{\sx}$ allows us to decompose it into the difference of two positive parts:
\begin{equation*}
    \small
    \begin{split}
        i_\cap^{\sx} (t:\mathbf{a}_1, ..., \mathbf{a}_m)    &= \log_2\frac{\mathbb{P}\left(\mathfrak{t} \cap \bigcup_{i=1}^m \mathfrak{a}_i \right)}{\mathbb{P}(\mathfrak{t}) \mathbb{P}\left(\bigcup_{i=1}^m \mathfrak{a}_i \right)} = \log_2\frac{1}{ \mathbb{P}\left(\bigcup_{i=1}^m \mathfrak{a}_i \right)}\\ 
        &-  \log_2\frac{\mathbb{P}(\mathfrak{t})}{\mathbb{P}\left(\mathfrak{t} \cap \bigcup_{i=1}^m \mathfrak{a}_i \right)},
    \end{split}
\end{equation*}
using standard rules for the logarithm. We call 
\begin{equation*}
    \small
    i_\cap^{\sx+}(t : \mathbf{a}_1, \dots, \mathbf{a}_m) := \log_2 \frac{1}{ \mathbb{P}\left(\bigcup_{i=1}^m \mathfrak{a}_i \right)}
\end{equation*}
the \emph{informative} local shared information and
\begin{equation*}
    \small
    i_\cap^{\sx-}(t : \mathbf{a}_1, \dots, \mathbf{a}_m) := \log_2\frac{\mathbb{P}(\mathfrak{t})}{\mathbb{P}\left(\mathfrak{t} \cap \bigcup_{i=1}^m \mathfrak{a}_i \right)}
\end{equation*}
the \emph{misinformative} local shared information.

\begin{acknowledgments}
We would like to thank Nils Bertschinger, Joe Lizier, Conor Finn and Robin Ince for fruitful discussions on PID.  We would also like to thank Patricia Wollstadt, Viola Priesemann, Raul Vicente, Johannes Zierenberg, Lucas Rudelt and Fabian Mikulasch for their valuable comments on this paper.

MW received support from SFB Project No.~1193, Subproject No.~C04 funded by the Deutsche Forschungsgemeinschaft. MW, AM, and AG are employed at the Campus Institute for Dynamics of Biological Networks (CIDBN) funded by the Volkswagen Stiftung. MW and AM received support from the Volkswagenstiftung under the program ``Big Data in den Lebenswissenschaften''. This work was supported by a funding from the Ministry for Science and Education of Lower Saxony and the Volkswagen Foundation through the ``Nieders\"{a}chsisches Vorab.'' MW is grateful to Jürgen Jost for hosting him at his department at the Max Planck Institute for Mathematics in the Sciences in Leipzig for a research stay funded by the Max Planck Society. 
\end{acknowledgments}

\providecommand{\noopsort}[1]{}\providecommand{\singleletter}[1]{#1}%

\end{document}